\documentclass[preprint,showpacs,preprintnumbers,amsmath,amssymb]{revtex4}

\usepackage{graphicx}
\usepackage{dcolumn}
\usepackage{bm}
\usepackage{amsmath}
\usepackage{amssymb}
\usepackage{mathrsfs}
\usepackage{color}

\newtheorem{theorem}{{\bf Theorem}}

\newtheorem{lemma}{{\bf Lemma}}
\newtheorem{definition}{{\bf Definition}}
\newtheorem{remark}{{\bf Remark}}

\newenvironment{proof}[1][Proof]{\begin{trivlist}
\item[\hskip \labelsep {\bfseries #1}]}{\end{trivlist}}
\newcommand{\qed}{\nobreak \ifvmode \relax \else
      \ifdim\lastskip<1.5em \hskip-\lastskip
      \hskip1.5em plus0em minus0.5em \fi \nobreak
      \vrule height0.75em width0.5em depth0.25em\fi}
      
\def \rmd{\mathrm{d}}

\begin{document}
\title{Relativistic helicity and link in Minkowski space-time}
\author{
Z. Yoshida$^{1}$,
Y. Kawazura$^{1}$, and
T. Yokoyama$^{2}$
}

\affiliation{
$^{1}$Graduate School of Frontier Sciences, University of Tokyo, Kashiwanoha, Kashiwa, Chiba 277-8561, Japan
\\
$^{2}$Department of Mathematics, Hokkaido University, Kita 10, Nishi 8, Kita-Ku, Sapporo, Hokkaido 060-0810, Japan
}

\begin{abstract}
A relativistic helicity has been formulated in the four-dimensional Minkowski space-time.
Whereas the relativistic distortion of space-time violates the conservation of the conventional helicity,
the newly defined relativistic helicity conserves in
a barotropic fluid or plasma,
dictating a fundamental topological constraint.
The relation between the helicity and the vortex-line topology has been
delineated by analyzing the linking number of vortex filaments
which are singular differential forms representing the pure states of Banach algebra.
While the dimension of space-time is four, vortex filaments link, because vorticities are primarily
2-forms and the corresponding 2-chains link in four dimension;
the relativistic helicity measures the linking number of vortex filaments that are
proper-time cross-sections of the vorticity 2-chains.
A thermodynamic force yields an additional term in the vorticity,
by which the vortex filaments on a reference-time plane are no longer pure states.
However, the vortex filaments on a proper-time plane remain to be pure states, if 
the thermodynamic force is exact (barotropic),
thus, the linking number of vortex filaments conserves.
\end{abstract}

\date{\today}
\maketitle

\newpage
\section{Introduction}
\label{sec:introduction}
The helicity of a vector field is a topological index
measuring the \emph{link}, \emph{twist} and \emph{writhe} of the field lines.
Here we say `index' because it is invariant under the action of diffeomorphism groups
generated by some ideal dynamics.
Conventionally, the helicity is defined for three-dimensional vectors;
that the dimension of space is three is, in fact, essential to define a helicity
(because, as to be formulated in a general setting, the helicity is the integral of a 3-form).
In this work, however, we study four-dimensional vectors in the Minkowski space-time.
In the relativistic dynamics, the conventional helicity is no longer an invariant.
One might connect the non-conservation of the helicity
with the topological fact that loops (1-cycles) do not link in four-dimensional space\,\cite{link}.
However, this is not true;
we will show that the field lines obeying an ideal (barotropic) equation of motion is still subject to the topological constraint;
two field lines do link in the spatial subspaces (temporal cross-sections of the space-time),
and the linking number conserves.
We will formulate a Lorentz-covariant helicity, show its invariance, and delineate its relation to the linking number.
On the other hand, we will find 
the reason why the conventional helicity fails to describe the link in the relativistic space-time.

We start with a short review of the conventional helicity.
Let $\bm{a}$ be a three-dimensional vector field defined on a domain $\Omega\subseteq\mathbb{R}^3$.
We call
\begin{equation}
\bm{b}=\nabla\times\bm{a}
\label{vorticity}
\end{equation}
the \emph{vorticity} of $\bm{a}$.
If $\Omega$ has a boundary $\partial\Omega$, we assume
\begin{equation}
\bm{n}\cdot\bm{b}=0,
\label{BC}
\end{equation}
where $\bm{n}$ is the unit normal vector onto $\partial\Omega$.
We define (assuming integrability)
\begin{equation}
C = 
\int_\Omega \bm{a}\cdot\bm{b}\,\rmd^3 x ,
\label{helicity_conventional}
\end{equation}
and call it the \emph{helicity} of $\bm{b}$ (or, sometimes, of $\bm{a}$) on $\Omega$;
the integrand is called the \emph{helicity density}.

The pioneering use of the helicity in the classical field theory was made by Woltjer\,\cite{Woltjer}
in order to characterize twisted magnetic field lines;
for the magnetic field $\bm{b}$ and its vector potential $\bm{a}$,
obeying the ideal magnetohydrodynamics (MHD) equation
(with the scalar potential $\phi$)
\begin{equation}
\partial_t \bm{a} = \bm{v}\times\bm{b} - \nabla \phi ,
\label{ideal_MHD}
\end{equation}
as well as the boundary condition (\ref{BC}),
the corresponding \emph{magnetic helicity} is a constant of motion.

For a fluid, we define the \emph{fluid helicity} by putting
$\bm{a}=\bm{v}$ (fluid velocity) and $\bm{b}=\bm{\omega}=\nabla\times\bm{v}$ (vorticity).
In an ideal fluid, $\bm{v}$ obeys the evolution equation
\begin{equation}
\partial_t \bm{v} = \bm{v}\times\bm{\omega} - \nabla \varepsilon ,
\label{ideal_fluid}
\end{equation}
where $\varepsilon$ is the total enthalpy
(we may write $\varepsilon=h+v^2/2$ with the static enthalpy $h$; 
in a homentropic flow, $\nabla h = n^{-1}\nabla p$
with the density $n$ and the pressure $p$).
If the boundary condition (\ref{BC}) holds, the fluid helicity is conserved.

The helicity is a topological index characterizing the twists of the bundle of field lines;
since it is an integral over a volume, some different geometrical characteristics of
field lines are summed up in $C$; see\,\cite{Moffatt,Moffatt-Ricca}.
To delineate its topological meaning in the simplest form,
let us consider a pair of \emph{vortex filaments} 
(to be identified as \emph{pure-state} differential forms; see Sec.\,\ref{subsec:pure-state}).
Suppose that $\Gamma_1$ and $\Gamma_2$ are a pair of disjoint loops (closed curves bounding disks) 
of class $C^1$ in $\mathbb{R}^3$.
The unit tangent vectors on these curves are denoted by $\bm{\ell}_1$ and $\bm{\ell}_2$, respectively.
The directions of $\bm{\ell}_1$ and $\bm{\ell}_2$ are arbitrarily chosen, 
and they are attributed to the loops $\Gamma_1$ and $\Gamma_2$ as their \emph{orientations}.
We regard both $\bm{\ell}_1$ and $\bm{\ell}_2$ as $\delta$-measures on the loops $\Gamma_1$ and $\Gamma_2$,
and consider a pair $\bm{b} = \bm{\ell}_1 + \bm{\ell}_2$.
For an arbitrary smooth loop $\Gamma$ that singly links with either $\Gamma_1$ or $\Gamma_2$,
we find, using Stokes's formula
(in the generalized sense for the singular $\bm{b}$; 
see Remark\,\ref{remark:generalize_Stokes_formula}) 
and the definition $\nabla\times\bm{a}=\bm{b}$,
\begin{equation}
{L}(\Gamma,\Gamma_j)=\oint_\Gamma \bm{a}\cdot\bm{\ell}\, \rmd\tau = \int_S \bm{b}\cdot\bm{n}\, \rmd\sigma 
= \mathrm{sgn} (\bm{n}\cdot\bm{\ell}_j)
\quad (j=1,2),
\label{Stokes}
\end{equation}
where $\bm{\ell}$ is the unit tangential vector on $\Gamma$,
$\rmd\tau$ is the length element on $\Gamma$, $S$ is a surface bounded by $\Gamma$, 
$\bm{n}$ is the unit normal vector on $S$, and
$\rmd\sigma$ is the surface element on $S$.
If $\Gamma$ and $\Gamma_j$ do not link, ${L}(\Gamma,\Gamma_j)=0$.
We may consider a more complex link of $\Gamma$ with $\Gamma_j$;
by a homotopical deformation of the path $\Gamma$ of the integral (\ref{Stokes}),
we find that ${L}(\Gamma,\Gamma_j)$ evaluates the \emph{linking number} of the pair of
loops $\Gamma$ and $\Gamma_j$
(the sign of link is determined by the orientations of the loops). 
By the definition, we may write 
$\bm{b}\,\rmd^3 x = \bm{\ell}_1 \rmd\tau_1 + \bm{\ell}_2 \rmd\tau_2$.
Using (\ref{Stokes}), we obtain
\begin{eqnarray}
C = \int_{\mathbb{R}^3} \bm{a}\cdot\bm{b}\, \rmd^3 x
&=& \oint_{\Gamma_1} \bm{a}\cdot\bm{\ell}_1 \rmd\tau_1 + \oint_{\Gamma_2} \bm{a}\cdot\bm{\ell}_2 \rmd\tau_2
\nonumber
\\
&=& 2 {L}(\Gamma_1,\Gamma_2).
\label{Stokes2}
\end{eqnarray}

The linking number in (\ref{Stokes2}) can be evaluated by the \emph{Biot-Savart integral}.
The relation $\bm{b}=\nabla\times\bm{a}$ is inverted as
\begin{equation}
\bm{a}(\bm{x}) = 
\frac{1}{4 \pi} \int_{\mathbb{R}^3} \frac{\bm{b}(\bm{x}')\times(\bm{x} -\bm{x}')}
{|\bm{x} -\bm{x}' |^3 } \rmd^3 x'.
\label{Biot-Savart}
\end{equation}
Inserting $\bm{b}\rmd^3 x' = \bm{\ell}_1 \rmd\tau_1 + \bm{\ell}_2 \rmd\tau_2$
into (\ref{Biot-Savart}), and using the resultant $\bm{a}$ in (\ref{helicity_conventional}),
we obtain
\begin{equation}
C= 2 \times \frac{1}{4 \pi} \oint_{\Gamma_1}\oint_{\Gamma_2} 
\frac{(\bm{x}_1 -\bm{x}_2)\cdot \bm{\ell}_1\rmd\tau_1\times\bm{\ell}_{2}\rmd\tau_2}
{|\bm{x}_1 -\bm{x}_2 |^3 }.
\label{Gauss-1}
\end{equation}
The right-hand-side integral ($= {L}(\Gamma_1,\Gamma_2)$) is called the \emph{Gauss linking number}.
The relation (\ref{Stokes2}) derived by Stokes' formula gives the proof that
${L}(\Gamma_1,\Gamma_2)$ evaluates an integer number counting the link of
two loops $\Gamma_1$ and $\Gamma_2$
(see \cite{Cantarella-Parsley} for generalization to a higher dimension;
see \cite{{Borromean}} for the Gauss integral for three-component links).

We will build a topological theory of relativistic field lines around two new constructions;
the first is an appropriate \emph{relativistic helicity} in the four-dimensional Minkowski space-time,
and the second is the notion of \emph{pure-state vorticities}
by which the helicity reads as the linking number of vortex filaments.
In Sec.\,\ref{sec:preliminaries}, we will start by reviewing the basic equations
that describe the relativistic ideal (barotropic) dynamics of a plasma (charged fluid).
The helicity will be defined for the vorticity of the canonical momentum that
combines the mechanical momentum and the electromagnetic field.
In Sec.\,\ref{sec:R-helicity}, we will formulate the relativistic helicity, and show its conservation.
Section\,\ref{sec:topology} is devoted for the delineation of the topological implication of the helicity conservation.
To this end, we will consider the link of vortex filaments which are formally the
aforementioned $\delta$-measures on co-moving loops.
We will justify them as the pure-states of Banach algebra, and show that they are the 
generalized (weak) solutions of the equation of motion.
For a pair of pure-state vortex filaments, the relativistic helicity 
evaluates their linking number in the spatial subspace, 
and its conservation parallels the relativistically corrected Kelvin's circulation law.

\begin{remark}[generalized Stokes' formula]
\label{remark:generalize_Stokes_formula}
To derive (\ref{Stokes2}), we evaluated the integral
$\oint_{\Gamma_j} \bm{a}\cdot\bm{\ell}_j \rmd\tau_j$ 
(so called \emph{circulation}; see Lemma\,\ref{lemma:circulation_law}) 
along the loop $\Gamma_j$ on which the $\delta$-measure vorticity $\bm{\ell}_j$ is supported,
and picked up the contribution from the other $\delta$-measure vorticity on the loop $\Gamma_k$ ($k\neq j$). 
However, the integrand $\bm{a}$ is not a continuous function, 
because it is generated by the $\delta$-measure $\bm{b}$.
Here are two mathematical issues pertinent to the use of Stokes' formula:
Let us decompose $\bm{a} = \bm{a}_1 +\bm{a}_2$ with
$\nabla\times\bm{a}_j = \bm{\ell}_j$ ($j=1,2$).  
\begin{enumerate}
\item
On the loop $\Gamma_1$, $\bm{a}_2$ is a smooth function,
thus $\oint_{\Gamma_1} \bm{a}_2 \cdot\bm{\ell}_1\rmd\tau_1$ can be
evaluated in the classical sense.
To relate this integral with the ``source'' $\bm{\ell}_2$ of $\bm{a}_2$, however, we
invoke Stokes's formula (\ref{Stokes}) in the generalized sense.
To justify (\ref{Stokes}) for a $\delta$-measure vorticity $\bm{b}=\bm{\ell}_2$,
we first approximate $\bm{\ell}_2$ by a smooth vector field $\bm{\omega}_\epsilon$ 
that gives the same surface integral $\int_{S_1} \bm{\omega}_\epsilon \cdot\bm{n}\,\rmd\sigma$
independent of $\epsilon$,
and path the limit of $\bm{\omega}_\epsilon \rightarrow \bm{\ell}_2$ 
($\epsilon\rightarrow0$) to define
\[
\int_{S_1} \bm{\ell}_2 \cdot\bm{n}\,\rmd\sigma 
= \lim_{\epsilon\rightarrow0} \int_{S_1} \bm{\omega}_\epsilon \cdot\bm{n}\rmd\sigma
= {L}(\Gamma_1,\Gamma_2).
\]

\item
To obtain (\ref{Stokes2}), we estimated
$\oint_{\Gamma_j} \bm{a}_j\cdot\bm{\ell}_j\rmd \tau_j =0$ ($j=1,2$),
which means that ${L}(\Gamma_j,\Gamma_j)=0$.  
In the neighborhood of each loop $\Gamma_j$, however, $\bm{a}_j$ is not continuous.
To justify these integrals, we consider a homotopy sequence of loops 
$\Gamma_\epsilon\rightarrow \Gamma_j$, and define
\[
\oint_{\Gamma_j} \bm{a}_j\cdot\bm{\ell}_j\rmd \tau_j = \lim_{\epsilon\rightarrow0}
\oint_{\Gamma_\epsilon} \bm{a}_j\cdot\bm{\ell}_\epsilon\rmd \tau_\epsilon = 0.
\]

\end{enumerate}

\end{remark}

\section{Preliminaries}
\label{sec:preliminaries}

\subsection{Basic Definitions}
\label{subsec:basic_definitions}
We denote the Minkowski space-time by $M\cong\mathbb{R}^4$.
On a reference frame, we may decompose $M=T\times X$ with $T\cong\mathbb{R}$ (time) and
$X\cong\mathbb{R}^3$ (space).
Following the standard notation, we write
\[
x^\mu=(ct,x,y,z), \quad x_\mu=(ct,-x,-y,-z),
\]
where $c$ is the speed of light. 
By a metric tensor
\begin{eqnarray}
	g^{\mu\nu} = 
  \left(
    \begin{array}{cccc}
			1 &  0 &  0 &  0\\
			0 & -1 &  0 &  0\\
			0 &  0 & -1 &  0\\
			0 &  0 &  0 & -1
    \end{array}
  \right)
	\label{metric}, 
\end{eqnarray}
we can write $x^\mu = g^{\mu\nu}x_\nu$.
The space-time gradients are denoted by
\[
\partial_\mu = \frac{\partial}{\partial x^\mu} = \left(\frac{\partial}{c\partial t}, \nabla\right),
\quad
\partial^\mu = \frac{\partial}{\partial x_\mu} = \left(\frac{\partial}{c\partial t}, -\nabla\right).
\]
The relativistic 4-velocity (normalized by $c$) is defined by the proper-time derivative:
\begin{equation}
{U}^\mu = \frac{\rmd x^\mu}{\rmd s} =(\gamma , \gamma\bm{v}/c), \quad
{U}_\mu = \frac{\rmd x_\mu}{\rmd s} =(\gamma , -\gamma\bm{v}/c),
\label{relativistic_4-velocity}
\end{equation}
where $\rmd s^2=\rmd x^\mu \rmd x_\mu$ and $\gamma=1/\sqrt{1-v^2/c^2}$.  
Obviously, ${U}^\mu {U}_\mu =1$.

The 4-momentum of a particle is
$ mc\,{U}^\mu  =(\gamma m c, \gamma m \bm{v})$,
where $m$ is the rest mass of the particle.
For a fluid, the \emph{effective rest mass} is given by $h/c^2$
with a proper (static) molar enthalpy $h$.
We may write, on a rest frame,
\[
h= \varepsilon + n^{-1}p,
\]
where $n$ is the number density, $p$ is the pressure, and 
$\varepsilon$ is the internal energy that includes the rest mass energy $mc^2$ as well as the thermal energy\,\cite{LandauLifshitz}.
The fluid 4-momentum is
\begin{equation}
{P}^\mu = (h/c)\,{U}^\mu =(\gamma h/c , \gamma (h/c^2) \bm{v}).
\label{fluid-4-momentum}
\end{equation}
Obviously, 
\begin{equation}
c U_\mu{P}^\mu = h.
\label{enthalpy}
\end{equation}

\begin{remark}[non-relativistic limit]
\label{remark:NR}
In the non-relativistic (NR) limit ($\gamma\rightarrow 1$), the 4-velocity 
(\ref{relativistic_4-velocity})
coincides with the reference-frame velocity:
\begin{equation}
u^\mu=c^{-1}\frac{\rmd x^\mu}{\rmd t} = (1, \bm{v}/c), \quad 
u_\mu=c^{-1}\frac{\rmd x_\mu}{\rmd t} = (1, -\bm{v}/c).
\label{NR_4-velocity}
\end{equation}
The NR particle 4-momentum is
$(mv^2/(2c), m\bm{v})$
(in the 0-component, we have subtracted the rest-mass energy $mc^2$ from the energy:
$\gamma m c^2 \rightarrow mc^2 + mv^2/2$).
The NR fluid momentum 4-vector is
\begin{equation}
P_{NR}^\mu =  (H/c, m\bm{v}),
\label{fluid-4-momentum-NR}
\end{equation}
where
\begin{equation}
H = \frac{mv^2}{2} + \varepsilon + n^{-1}p = \frac{mv^2}{2} + h.
\label{fluid-4-momentum-NR'}
\end{equation}
We need to assume $mv^2/2\ll h$ to approximate
$cu_\mu P_{NR}^\mu = h-mv^2/2 \approx h$.
\end{remark}

\subsection{Field tensor and equation of motion}
\label{subsec:field-tensor}
The energy-momentum tensor of a fluid is
\[
T_{\mu\nu}= {n h} U_\mu U_\nu - p g_{\mu\nu}.
\]
The quasi-static (entropy-conserving) equation of motion is derived by 
$\partial^\nu T_{\mu\nu}=0$, which reads 
\begin{equation}
c{U}^\mu {M}_{\mu\nu} + \partial_\nu h - n^{-1} \partial_\nu p = 0 ,
\label{EQM-tensor}
\end{equation}
where 
\[
{M}_{\mu\nu} = \partial_\mu {P}_\nu - \partial_\nu {P}_\mu.
\]
is the \emph{matter field tensor}.
By the thermodynamic first and second laws, we may write,
for an isentropic flow,
\begin{equation}
\partial_\nu h = n^{-1} \partial_\nu p + T \partial_\nu S
\label{thermodynamic1+2}
\end{equation}
with the entropy $S$.  Hence, the equation of motion (\ref{EQM-tensor}) is rewritten as\,\cite{LandauLifshitz}
\begin{equation}
{U}^\mu {M}_{\mu\nu}  = -c^{-1} T \partial_\nu S . 
\label{EQM-tensor'}
\end{equation}
Contracting both sides of (\ref{EQM-tensor'}) with $U^\nu$, we obtain
$T U^\nu \partial_\nu S=0$, implying the entropy conservation
($\partial_t S + \bm{v}\cdot\nabla S=0$ on a reference frame).

\subsection{Barotropic fluid}
\label{subsec:barotropic}

In the present work, \emph{barotropic fluids} will be 
at the center of discussions;
when $S$ is a function of $T$, we may write $T\rmd S = \rmd \theta$,
and then, (\ref{EQM-tensor'}) reads
\begin{equation}
{U}^\mu {M}_{\mu\nu}  = - c^{-1} \partial_\nu \theta . 
\label{EQM-tensor-barotropic}
\end{equation}
Evidently, $U^\nu \partial_\nu \theta=0$.

Notice that the heat term $T\rmd S$ is reduced into an exact differential $\rmd\theta$.
This is the root cause of various topological constraints in a barotropic flow.
For example, Kelvin's circulation theorem (see Lemma\,\ref{lemma:circulation_law}) is
a consequence of vanishing heat cycle $\oint T\rmd S = \oint\rmd\theta=0$
(see \cite{MahajanYoshida2010} for the relation between the circulation law and
heat cycles, as well as its relativistic generalization).
However, the exactness of $\rmd\theta$ is now in the relativistic space-time;
its representation on a reference frame may be non-exact. 
This difference causes an interesting ``relativistic effect'' which is, indeed, the
target of the present exploration.  

Henceforth, we will assume a barotropic relation, and discuss the equation of motion in
the form of (\ref{EQM-tensor-barotropic});
one may generalize the following formulations to a \emph{baroclinic} fluid by replacing
$\rmd\theta$ by $T\rmd S$.

\begin{remark}[NR barotropic fluid]
\label{remark:NR-barotropic}
It is remarkable that, in the NR limit, the barotropic heat term $\partial^\nu \theta$ can be
absorbed by the energy term (0-component) of the fluid 4-momentum:
modifying (\ref{fluid-4-momentum-NR'}) as
\begin{equation}
H =  \frac{mv^2}{2} + h + \theta,
\label{fluid-4-momentum-NR''}
\end{equation}
and redefining the field tensor ${M}$ by the modified 4-momentum, the
NR equation of motion (\ref{EQM-tensor-barotropic}) can be reduced to
\begin{equation}
{u}^\mu {M}_{\mu\nu}  = 0 . 
\label{EQM-tensor-barotropic'}
\end{equation}
Hence, a NR barotropic fluid equation is equivalent to a homentropic one ($\rmd S=0$).
However, such reduction is not possible in the relativistic regime.    
\end{remark}

\subsection{Charged fluid (plasma)}
\label{subsec:plasma}
For a charged fluid (plasma), we have to 
dress the momentum 4-vector with
the electromagnetic (EM) potential $A^\mu=(\varphi, \bm{A})$, and replace the 4-momentum by the
\emph{canonical 4-momentum}
\[
\mathcal{P} = {P}+q A,
\]
where $q$ is the charge.
The matter-EM field tensor (2-form) is\,\cite{Mahajan2003}
\[
{\mathcal{M}}_{\mu\nu} 
=(\partial_\mu {P}_\nu - \partial_\nu {P}_\mu) +
q(\partial_\mu {A}_\nu - \partial_\nu {A}_\mu) ,
\]
by which the equation of motion (\ref{EQM-tensor-barotropic}) modifies to include the
Lorentz force as
\begin{equation}
{U}^\mu \mathcal{M}_{\mu\nu}  
=-c^{-1}\partial_\nu\theta .
\label{EQM-tensor-plasma}
\end{equation}
The relation (\ref{enthalpy}) of a neutral-fluid is generalized as
\begin{equation}
U^\mu\mathcal{P}_\mu = c^{-1}(h + q \varrho)
\quad (\varrho= U_\mu A^\mu).
\label{enthalpy2}
\end{equation}

Henceforth, we will consider the EM-dressed (or canonical) momentum $\mathcal{P}^\mu$
and its derivatives (field tensor, helicity, etc.).
Synonymously, we may call it a matter-dressed EM potential, i.e.,
the EM potential $A^\mu$ (multiplied by the charge $q$) coupled with the matter motion $P^\mu$,
which will be denote by $\mathcal{A}^\mu$.
The latter EM-based notation will be useful when we compare newly defined helicities
with the familiar magnetic helicity (see Sec.\,\ref{subsec:tensor}).
The conventional MHD equation (\ref{ideal_MHD}) is the zero mass and heat
($P=0$ and $\theta=0$) limit of the spatial component of (\ref{EQM-tensor-plasma}).

\subsection{Representations on a reference frame}
\label{subsec:tensor}

For the convenience of the forthcoming calculations,
we write down the components of the field tensors
and the equations of motion on a reference frame.
Here we invoke an analogy of the familiar Faraday tensor to denote the
elements of the field tensors.
We write the spacial components in bold-face symbols.
Denoting the canonical momentum (or the dressed EM potential) by
\begin{equation}
\mathcal{P}^\mu 
=(\mathcal{P}^0,\bm{\mathcal{P}})
\equiv \mathcal{A}^\mu = (\mathcal{A}^0,\bm{\mathcal{A}}),
\label{P-components}
\end{equation}
we define
\begin{eqnarray*}
& & \bm{\mathcal{E}} = - \nabla \mathcal{A}^0 - c^{-1}\partial_t \bm{\mathcal{A}} ,
\\
& & \bm{\mathcal{B}} = \nabla\times \bm{\mathcal{A}}.
\end{eqnarray*}
Then, the field tensors are written as
\begin{eqnarray}
\mathcal{M}^{\mu\nu} &=&
\partial^\mu \mathcal{P}^\nu - \partial^\nu \mathcal{P}^\mu
= \left( \begin{array}{cccc}
0             & -\mathcal{E}_1 & -\mathcal{E}_2 & -\mathcal{E}_3 \\
\mathcal{E}_1 &  0             & -\mathcal{B}_3 &  \mathcal{B}_2 \\
\mathcal{E}_2 &  \mathcal{B}_3 &  0             & -\mathcal{B}_1 \\
\mathcal{E}_3 & -\mathcal{B}_2 &  \mathcal{B}_1 &  0             \\
\end{array} \right) ,
\label{M-components_up}
\\
\mathcal{M}_{\mu\nu} &=&
\partial_\mu \mathcal{P}_\nu - \partial_\nu \mathcal{P}_\mu
= \left( \begin{array}{cccc}
 0             &  \mathcal{E}_1 &  \mathcal{E}_2 &  \mathcal{E}_3 \\
-\mathcal{E}_1 &  0             & -\mathcal{B}_3 &  \mathcal{B}_2 \\
-\mathcal{E}_2 &  \mathcal{B}_3 &  0             & -\mathcal{B}_1 \\
-\mathcal{E}_3 & -\mathcal{B}_2 &  \mathcal{B}_1 &  0             \\
\end{array} \right) ,
\label{M-components_down}
\\
\mathcal{M}^{*\mu\nu} &=&
\frac{1}{2}\epsilon^{\mu\nu\alpha\beta} \mathcal{M}_{\alpha\beta} 
= \left( \begin{array}{cccc}
 0            & -\mathcal{B}_1 & -\mathcal{B}_2 & -\mathcal{B}_3 \\
\mathcal{B}_1 &  0             &  \mathcal{E}_3 & -\mathcal{E}_2 \\
\mathcal{B}_2 & -\mathcal{E}_3 &  0             &  \mathcal{E}_1 \\
\mathcal{B}_3 &  \mathcal{E}_2 & -\mathcal{E}_1 &  0             \\
\end{array} \right) .
\label{*M-components_down}
\end{eqnarray}
We may identify $\mathcal{M}^{*\mu\nu}$ as the Minkowski-Hodge dual of $\mathcal{M}_{\mu\nu}$;
see Remark\,\ref{remark:MH-dual}.
We have an obvious relation
\begin{equation}
	\partial_\mu\mathcal{M}^{*\mu\nu} 
= (\nabla\cdot\bm{\mathcal{B}},\, -c^{-1}\partial_t \bm{\mathcal{B}} - \nabla\times\bm{\mathcal{E}}) = 0.
\label{div Mstar}
\end{equation}
The equation of motion (\ref{EQM-tensor-barotropic}) reads
\begin{eqnarray}
& & \bm{v}\cdot\bm{\mathcal{E}}
=  (c\gamma)^{-1} \partial_t \theta , 
\label{EQM-vector-0}
\\
& & \bm{\mathcal{E}}+(\bm{v}/c)\times\bm{\mathcal{B}} 
= -(c\gamma)^{-1} \nabla\theta .
\label{EQM-vector-1}
\end{eqnarray}

\subsection{Differential forms}
\label{subsec:forms}

In the following analysis, the differential-geometric formulations
of the field tensors and the equation of motion will play an essential role.
The 4-velocity $U^\mu$ is regarded as a vector field $U=U^\mu\partial_\mu
\in T M$ ($M$ is the Minkwski space-time, and $TM$ is the tangent bundle on $M$).
We denote the corresponding covector (1-form) by
$\mathcal{U}=U_\mu \rmd x^\mu \in T^* M$ (the cotangent bundle on $M$).
The 4-momentum is a 1-form: $\mathcal{P}=\mathcal{P}_\mu\rmd x^\mu \in T^* M$, 
and its exterior derivative is the field tensor:
\begin{eqnarray}
\mathcal{M}=\rmd \mathcal{P}
&=& \partial_\mu \mathcal{P}_\nu \rmd x^\mu\wedge\rmd x^\nu
\label{field_tensor} \\
&=& \frac{1}{2} \mathcal{M}_{\mu\nu} \rmd x^\mu\wedge\rmd x^\nu .
\label{field_tensor'}
\end{eqnarray}
The components of $*\mathcal{M}$ (the Minkowski-Hodge-dual of $\mathcal{M}$) will be denoted by $\mathcal{M}^*_{\alpha\beta}$,
i.e. $*\mathcal{M}= (1/2)\mathcal{M}^{*\alpha\beta}\rmd x_\alpha\wedge\rmd x_\beta$;
see (\ref{*M-components_down}).
Using the definition (see Remark\,\ref{remark:MH-dual})
\begin{equation}
	\frac{1}{2}\epsilon^{\alpha\beta\mu\nu}\, \rmd x_\alpha \wedge \rmd x_\beta = *(\rmd x^\mu \wedge \rmd x^\nu) ,
\label{2-form basis identity}
\end{equation}
we may write
\begin{equation}
	*\mathcal{M}= \frac{1}{2}\mathcal{M}_{\mu\nu}\, *(\rmd x^\mu\wedge\rmd x^\nu).
\label{*M}
\end{equation}

Denoting by $i_U$ the interior product with $U$, we write
$i_U\mathcal{M}=U^\mu\mathcal{M}_{\mu\nu}\rmd x^\nu$.
The equation of motion (\ref{EQM-tensor-plasma}) is now written as
\begin{equation}
i_U\mathcal{M} =-c^{-1}\rmd \theta . 
\label{EQM-dif-form}
\end{equation}

Operating the Lie derivative $L_U = i_U \rmd + \rmd i_U$ on the momentum 1-form $\mathcal{P}$,
we obtain
\begin{equation}
L_U \mathcal{P} = i_U \rmd \mathcal{P} + \rmd i_U\mathcal{P}
= i_U \mathcal{M} + c^{-1}\rmd (h+q\varrho),
\label{Lie-derivative}
\end{equation}
where we have used the relation (\ref{enthalpy2}).
Hence (\ref{EQM-dif-form}) reads
\begin{equation}
L_U \mathcal{P} 
= c^{-1}\rmd (h+q\varrho-\theta).
\label{EQM-Lie-form}
\end{equation}
The right-hand side is the exact differential of the \emph{free energy}.
In the following discussion, we will use the form (\ref{EQM-Lie-form}) of the equation of motion.

The 2-form $\mathcal{M}=\rmd\mathcal{P}$ may be called the ``four-dimensional vorticity'' of the canonical momentum.
The relativistic vorticity equation is derived from (\ref{EQM-dif-form}):
\begin{equation}
 L_U\mathcal{M} = \rmd i_U \mathcal{M} = 0 ,
\label{VD-Lie-form}
\end{equation}

\begin{remark}[Minkowski-Hodge dual] 
\label{remark:MH-dual}
Combining the Hodge-duality (between $p$-forms and 
$(n-p)$-forms)
and the Minkowski-duality (between $\rmd x^\mu$ and $\rmd x_\mu$), 
the ``star operator'' maps the basis as
\begin{eqnarray*}
& &
*(\rmd x^{i_1} \wedge \cdots \wedge \rmd x^{i_p}) 
\\
& &~~~= \frac{1}{(n - p)!}\epsilon^{i_1 \cdots i_p}_{\hspace{2.5em}j_1 \cdots j_{n-p}} g^{j_1 k_1} \cdots g^{j_{(n-p)} k_{(n-p)}} \rmd x_{k_1} \wedge \cdots \wedge \rmd x_{k_{(n-p)}}.
\end{eqnarray*}
For example, if $\mathcal{M}$ is a 2-form based on $\rmd x^\mu\wedge\rmd x^\nu$,
the Minkowski-Hodge dual $*\mathcal{M}$ is a 2-form  based on $\rmd x_\mu\wedge\rmd x_\nu$;
see (\ref{*M-components_down}), (\ref{2-form basis identity})
and (\ref{*M}).

\end{remark}

\subsection{Diffeomorphism generated by 4-velocity}
\label{subsec:diffeo}

The ideal fluid/plasma motion is represented by a diffeomorphism $\mathcal{T}_U(s)$
that is generated by a vector $U$ on the Minkowski space-time $M\cong \mathbb{R}^4$, i.e.,
\begin{equation}
\frac{\rmd}{\rmd s} \mathcal{T}_U(s) = U.
\label{diffeomorphism}
\end{equation}
For every initial point $x^\mu(0)\in M$, 
$\mathcal{T}_U(s) x^\mu(0) = x^\mu(s)$ and 
$\rmd x^\mu(s)/\rmd s = U^\mu$.
We assume that $\{\mathcal{T}_U(s);\,s\in\mathbb{R}\}$ is a
one-parameter group of $C^2$-class diffeomorphisms that are
continuously differentiable with respect to $s$.

By $\mathcal{T}_U(s)$, we can define various geometrical objects (chains) ``co-moving'' with the
fluid.  
For the latter convenience, we introduce two different temporal cross-sections
of the space-time $M$ (which we call planes, although they have three dimensions).
The ``$t$-plane'' is, for a fixed parameter $t\in\mathbb{R}$,
\begin{equation}
\Xi(t)=\{(x^0, x^1,x^2,x^3);\, x^0=ct,\, (x^1,x^2,x^3)\in X \}.
\label{t-plane}
\end{equation}
In the relativistic theory, 
the proper time $s$-plane plays a more essential role,
which is a three-dimensional hyper-surface in $M$ defined by
\begin{equation}
\tilde{\Xi}(s) = \mathcal{T}_U(s) \Xi(0).
\label{s-plane}
\end{equation}
Evey object that is initially contained in $\Xi(0)$ 
stays in $\tilde{\Xi}(s)$ ---such an object is said \emph{co-moving}.
We note that a co-moving object does not stay on a $t$-plane $\Xi(t)$;
the {synchroneity} (simultaneity) is broken when $\gamma$ is inhomogeneous.
A co-moving object $V(s)\subset\tilde{\Xi}(s)$ is \emph{temporally-thin}, i.e.,
$V(s)\cap V(s')=\emptyset$ for every $s\neq s'$,
or $V(s)$ has zero measure with respect to $x^0$.
This is because the diffeomorphism $\mathcal{T}_U(s)$ is {dynamical}
in the sense that $U^0 > 0$, and thus,
$\mathcal{T}_U(s)$ ($s\neq0$) does not have any fixed point in $M$.

Before starting analyses, we make a remark on the framework of our discussions.
The vector $U$, representing the velocity of fluid elements, must be consistent
with the evolution of other physical quantities
$\mathcal{P}$, $h$, and $\theta$.
We do not know, however, a theorem guaranteeing the global existence 
of solutions to the nonlinear system of equations;
it is not possible to construct self-consistent fields $U$, $\mathcal{P}$, $\mathcal{M}$, etc. 
by solving the evolution equations.  
Yet, we may discuss \emph{a priori} properties of the
solutions, i.e., the mathematical relations that must be satisfied by every solution
(of certain class) whenever it exists.
The aim of this study is to derive a topological conservation law pertinent to
the diffeomorphism group $\{\mathcal{T}_U(s);\,s\in\mathbb{R}\}$.
The helicity will guide our exploration.

\section{Relativistic helicity in the Minkowski space-time}
\label{sec:R-helicity}

\subsection{Semi-relativistic helicity}
\label{subsec:semi-relativistic}
Conventionally, a helicity is defined by a three-dimensional integral
$C = \int_X \bm{a}\cdot\bm{b}\,\rmd^3 x$
with three-dimensional vectors $\bm{a}$ and $\bm{b}=\nabla\times\bm{a}$ defined on space $X$; 
see (\ref{helicity_conventional}).
Here, we consider the helicity of ``matter-dressed EM field''
(or, synonymously,  ``EM-dressed fluid momentum'') by inserting $\bm{\mathcal{A}}$ ($=\bm{\mathcal{P}}$)
into $\bm{a}$, and $\bm{\mathcal{B}}$ into $\bm{b}$;
see Sec.\,\ref{subsec:tensor}.
By straightforward extension of the well-known magnetic-helicity conservation
in an ideal fluid obeying (\ref{ideal_MHD}) (or the fluid-helicity conservation by (\ref{ideal_fluid})), we find that $C=\int_X \bm{\mathcal{A}}\cdot\bm{\mathcal{B}}\,\rmd^3x$ is conserved in a non-relativistic 
barotropic charged fluid (plasma).
In relativistic dynamics, however, $C$ is no longer a constant;
the relativistic effect yields a ``relativistic baroclinic effect'' 
by space-time distortion\,\cite{MahajanYoshida2010}.
The aim of the latter discussion is 
to formulate a more appropriate helicity that conserves in relativistic motion, and
use it to elucidate a fundamental topological property of relativistic dynamics.

Our construction starts by relating the helicity to a 3-form such as
\begin{equation}
\mathcal{K} = \mathcal{P}\wedge \rmd \mathcal{P} .
\label{K}
\end{equation}
Explicitly, we may write
\begin{eqnarray}
	\mathcal{K} 
&=& \frac{1}{2}\mathcal{P}_\nu\mathcal{M}_{\alpha\beta}\rmd x^\nu \wedge \rmd x^\alpha \wedge \rmd x^\beta
\nonumber \\
&=& -{\mathcal{P}}_\nu \mathcal{M}^{*\mu\nu}\,*(\rmd x_\mu) ,
\label{K'}
\end{eqnarray}
where we used the identity $*(\epsilon^{\mu\nu\alpha\beta}\rmd x_\mu) = -\rmd x^\nu \wedge \rmd x^\alpha \wedge \rmd x^\beta$.
Inserting (\ref{P-components}) and (\ref{*M-components_down}), we may write
$\mathcal{K} = -\mathcal{K}^\mu\,*(\rmd x_\mu)$ with
\begin{eqnarray*}
	\mathcal{K}^0 & = & {\mathcal{P}}_\nu \mathcal{M}^{*0\nu} = \bm{\mathcal{A}}\cdot\bm{\mathcal{B}} , \\
	\mathcal{K}^j & = & {\mathcal{P}}_\nu \mathcal{M}^{*j\nu} = \left(\mathcal{A}_0\bm{\mathcal{B}} - \bm{\mathcal{A}}\times\bm{\mathcal{B}}\right)^j.
\end{eqnarray*}
Hence, we may write the conventional helicity as
\begin{equation}
C=\int_X \mathcal{K}^0 \rmd^3 x = \int_X {\mathcal{P}}_\nu \mathcal{M}^{*0\nu} \rmd^3 x.
\label{semi-classical-helicity}
\end{equation}
As the domain of integration (the $t$-plane $X$) is not Lorentz covariant,
$C$ depends on the choice of a reference frame,
Therefore, we call $C$ a \emph{semi-relativistic helicity}.
A fully relativistic helicity must be defined by
both covariant integrand and a domain of integration
---this will be done in the next subsection.

Here we examine $\rmd C/\rmd t$ explicitly.
Using (\ref{div Mstar}) and (\ref{EQM-vector-1}), 
we obtain
\begin{equation}
	\partial_\mu \mathcal{K}^\mu = {\cal M}_{\mu\nu}{\cal M}^{*\mu\nu} 
= -2 \bm{\mathcal{E}}\cdot\bm{\mathcal{B}} = 2 (c\gamma)^{-1} \bm{\mathcal{B}}\cdot\nabla \theta.
\label{e:div K}
\end{equation}
We thus find
\begin{eqnarray}
\frac{\rmd}{\rmd t}C 
&=& c\int_X \partial_0 \mathcal{K}^0 \rmd^3x
\nonumber \\
&=& c \int_X \left( \partial_\mu \mathcal{K}^\mu - \partial_j \mathcal{K}^j \right) \rmd^3x
\nonumber \\
&=& c \int_X \partial_\mu \mathcal{K}^\mu \rmd^3x
\nonumber \\
&=& 2 \int_X \gamma^{-1}\bm{\mathcal{B}}\cdot\nabla\theta\,\rmd^3 x
= -2 \int_X \theta\bm{\mathcal{B}}\cdot\nabla\gamma^{-1}\,\rmd^3 x.
\label{semi-classical-helicity-change}
\end{eqnarray}
Obviously $C$ is invariant in the non-relativistic limit ($\gamma=1$);
we now have a unified non-relativistic helicity conservation law that combines
the magnetic-helicity conservation and fluid-helicity conservation
introduced in Sec.\,\ref{sec:introduction}.
Interestingly, in a homentropic fluid ($\rmd\theta=0$), 
$\partial_\mu\mathcal{K}^\mu=0$,
thus, $C$ is constant regardless of the relativistic effect.
Otherwise, the semi-relativistic helicity $C$ is not a constant of motion.

\subsection{Relativistic helicity}
\label{subsec:R-helicity}

We define a \emph{relativistic helicity} $\mathfrak{C}$ in the Minkowski space-time by the integral of the 3-form $\mathcal{K}$ over
a co-moving three-dimensional volume $V(s) = \mathcal{T}_U(s)V_0$
($V_0 \subset \Xi(0)$; see Sec.\,\ref{subsec:diffeo}):
\begin{equation}
\mathfrak{C}(s) 
= \int_{V(s)} \mathcal{P}\wedge \rmd \mathcal{P} .
\label{4Dhelicity}
\end{equation}
Notice that the integrand 3-form $\mathcal{K}=\mathcal{P}\wedge \rmd \mathcal{P}$ includes
space-time coupled terms ($\sum_{jk\ell}\mathcal{K}_j\rmd x^0\wedge\rmd x^k\wedge\rmd x^\ell$),
and the synchroneity (simultaneity) may be broken on $V(s)$.

\begin{theorem}[relativistic helicity conservation]
\label{theorem:helicity_conservation}
Suppose that $\mathcal{P}~(\in C^2(M))$ is a
solution of (\ref{EQM-Lie-form}) such that $\mathrm{supp}\,\rmd\mathcal{P}\cap\Xi(0)\subsetneq V_0$ (a bounded set).
The helicity $\mathfrak{C}(s)$ 
evaluated on a co-moving domain $V(s)=\mathcal{T}_U(s)V_0$ is a constant of motion
(i.e. $\rmd\mathfrak{C}(s)/\rmd s=0$).
\end{theorem}

\begin{proof}
By the definition of $V(s)$, we obtain
\begin{eqnarray}
\frac{\rmd}{\rmd s} \mathfrak{C} &=& \frac{\rmd}{\rmd s} \int_{V(s)} \mathcal{P}\wedge \rmd \mathcal{P}
\nonumber 
\\
&=& \int_{V(s)} L_U (\mathcal{P}\wedge \rmd \mathcal{P}).
\label{4Dhelicity_conservation-1}
\end{eqnarray}
Operating by $\rmd$ on both sides of the equation of motion (\ref{EQM-Lie-form}), which reads
\[
i_U \rmd \mathcal{P} + \rmd i_U \mathcal{P}
= c^{-1}\rmd (h+q\varrho-\theta),
\]
we find $\rmd i_U \rmd \mathcal{P} = 0$.
Using this, we obtain the vorticity equation
\begin{equation}
L_U \rmd \mathcal{P} 
= (\rmd i_U +  i_U \rmd) \rmd \mathcal{P} 
= 0.
\label{vorticity_equation}
\end{equation}
The integrand of (\ref{4Dhelicity_conservation-1}) reads
\begin{eqnarray}
L_U (\mathcal{P}\wedge \rmd \mathcal{P}) 
&=& (L_U \mathcal{P})\wedge \rmd \mathcal{P} 
\nonumber \\
&=& c^{-1}\rmd (h+q\varrho-\theta)\wedge \rmd \mathcal{P}
\nonumber \\
&=& c^{-1}\rmd \left[ (h+q\varrho-\theta) \rmd \mathcal{P} \right].
\label{4Dhelicity_conservation_Lie-derivative}
\end{eqnarray}
We thus have
\begin{equation}
\frac{\rmd}{\rmd s} \mathfrak{C} 
= c^{-1}\int_{V(s)} \rmd [(h+q\varrho-\theta) \rmd \mathcal{P}] 
= c^{-1}\int_{\partial V(s)} (h+q\varrho -\theta) \rmd \mathcal{P}.
\label{4Dhelicity_conservation}
\end{equation}
If $\rmd \mathcal{P}=0$ on the boundary $\partial V(s)$, 
$\rmd\mathfrak{C}/\rmd s =0$.

Now we study the relation between $\mathrm{supp}\,\rmd\mathcal{P}$ and $V(s)$.
By the assumption, $\mathrm{supp}\,\rmd\mathcal{P}\cap\tilde{\Xi}(0) \subsetneq V(0)$.
The following Lemma\,\ref{lemma:support}, then, shows 
\begin{equation}
\mathrm{supp}\,\rmd\mathcal{P}\cap\tilde{\Xi}(s) \subsetneq V(s)
\quad (\forall s).
\label{support_of_M}
\end{equation}
We thus find that $\mathrm{supp}\,\rmd\mathcal{P}$ is
contained in a column that has compact $s$-cross-sections.
Since $V(s)$ is temporally thin (see Sec.\,\ref{subsec:diffeo}), 
$\partial V(s)$ is {purely spacial},
i.e. for every $\sigma\subset \partial V(s)$, $\int_\sigma \rmd x^0\wedge\rmd x^j=0$
($\forall j$).
On each $s$-plane $\tilde{\Xi}(s)$, $\rmd\mathcal{P} = 0$
on $\partial V(s)$.
\qed
\end{proof}

To complete the proof of Theorem\,\ref{theorem:helicity_conservation},
we have yet to prove Lemma\,\ref{lemma:support}.
We prepare the following Lemma:

\begin{lemma}[circulation law]
\label{lemma:circulation_law}
Let $\Gamma(s)=\mathcal{T}_U(s) \Gamma_0$ be
an arbitrary co-moving loop (1-chain bounding a disk) of class $C^1$.
In a tubular neighborhood of $\Gamma(s)$, we assume that 
a 1-form $\mathcal{P}$ is continuously differentiable and satisfies 
the equation of motion (\ref{EQM-Lie-form}).
Then, the \emph{circulation}
\begin{equation}
\Phi(s) = \oint_{\Gamma(s)} {\mathcal{P}}
\label{R-circulation}
\end{equation}
is conserved, i.e. $\rmd \Phi(s)/\rmd s =0$.
\end{lemma}

\begin{proof}
In Sec.\,\ref{sec:topology},
we will use this Lemma for singular solutions of (\ref{EQM-Lie-form}), so
we assume that $\mathcal{P}$ is a classical solution of (\ref{EQM-Lie-form})
only in a tubular neighborhood the loop $\Gamma(s)$
where we need to evaluate the circulation.
For the existence of a tubular neighborhood of a $C^1$-submanifold, see Hirsch\,\cite{link}.
The derivation is straightforward: 
\begin{equation}
\frac{\rmd}{\rmd s}\oint_{\Gamma(s)} {\mathcal{P}}
= \oint_{\Gamma(s)} L_U {\mathcal{P}} 
= c^{-1}\oint_{\Gamma(s)} \rmd (h + q\varrho -\theta)
=0.
\label{circulation}
\end{equation}
\qed
\end{proof}

It is often convenient to rewrite the circulation (\ref{R-circulation}) in terms of the
\emph{vorticity} $\rmd\mathcal{P}$ and a co-moving surface $\sigma(s)$
such that $\partial\sigma(s)=\Gamma(s)$;
by Stokes' formula, 
\begin{equation}
\Phi(s) = \int_{\partial\sigma(s)} \mathcal{P} = \int_{\sigma(s)}\rmd\mathcal{P} .
\label{R-circulation-2}
\end{equation}


\begin{lemma}[vortex motion]
\label{lemma:support}
Suppose that $\mathcal{P}~(\in C^1(M))$ is a
solution of the equation of motion (\ref{EQM-Lie-form}).
Let $\mathcal{M}=\rmd\mathcal{P}$, and $S_0=\mathrm{supp}\,\mathcal{M}\cap\Xi(0)$.
Then, 
\begin{equation}
\mathrm{supp}\,\mathcal{M}\cap\tilde{\Xi}(s) =\mathcal{T}_U(s) S_0
\quad (\forall s).
\label{support_of_M-2}
\end{equation}
\end{lemma}

\begin{proof}
By Lemma\,\ref{lemma:circulation_law} and (\ref{R-circulation-2}), 
any co-moving surface $\sigma(s)$ that
does not intersect with $S_0$ at $s=0$ yields zero circulation for every $s\in\mathbb{R}$.
On the other hand, if $\sigma(0)$ and $S_0$ intersect to yield a finite circulation,
it conserves through the motion of $\sigma(s)$.
Therefore, $\mathcal{M}\neq0$ only on $\sigma(s)$ that starts from $\sigma(0)$ containing a
point on $S_0$,
implying (\ref{support_of_M-2}). 
\qed
\end{proof}

\begin{remark}[Kelvin's circulation theorem and connection theorem]
\label{remark:Kelvin}
Lemma\,\ref{lemma:circulation_law} is the special-relativistic correction of
{Kelvin's circulation theorem}
(see \cite{Synge} for a general-relativistic treatment).
The point is that the circulation must be evaluated on a co-moving loop ${\Gamma(s)}$,
on which the {synchroneity} (simultaneity) with respect to a reference-frame time $t$ is broken;
if we evaluate the circulation on a synchronic loop $\Gamma(t)$, 
the corresponding \emph{semi-relativistic} circulation is not invariant\,\cite{MahajanYoshida2010}.
In Lemma\,\ref{lemma:support}, we used the circulation law to show that the vorticity co-moves with the fluid.
The same argument applies to show that every vortex line is frozen in the fluid element.
The motion of a vortex line  (or a magnetic field line in an MHD system)
can be described explicitly by the \emph{connection equation} that governs the separation vector connecting infinitesimally close fluid elements\,\cite{Newcomb}.
Pegoraro\,\cite{Pegoraro} formulated the Lorentz-covariant connection equation for a relativistic MHD system.
Interestingly, the \emph{magnetic flux} on a synchronic surface $\sigma(t)\subset\Xi(t)$
is invariant in the MHD system, because the right-hand side of (\ref{e:div K}) vanishes in the MHD model (see also Sec.\,\ref{sec:discussion}).
In parallel to this fact, the connection equation of MHD may be written in a seemingly 
non-covariant form.
Pegoraro's formulation considers a surface orbit in space-time to overcome the lack of simultaneity on the co-moving magnetic field line.
In the next section, we will consider a ``singular'' vortex surface in space-time, which, however, is
different from the virtual surface of the separation vectors;
it is physically the orbit in space-time of a vortex filament carrying a unit vorticity,
and is mathematically the pure state of Banach algebra.
\end{remark}

\section{Linking in the Minkowski space-time}
\label{sec:topology}

\subsection{Topological constraint by the helicity}
Link is definable (as a homotopy invariant) for a pair of geometric objects (chains) having codimension
less than or equal to one\,\cite{link}; 
for example, two loops (1-chains bounding disks) may link in three-dimensional space, while they do not in four-dimensional space.
Therefore, it seems that the helicity ceases to be related to linking numbers in the four-dimensional relativistic space-time.
However, the relativistic helicity conservation, derived in Sec.\,\ref{subsec:R-helicity},
does impose a topological constraint.
The aim of this section is to elucidate the topological meaning of
the relativistic helicity conservation by generalizing the conventional relation between the
helicity and the link of vortex filaments (see Sec.\,\ref{sec:introduction}) in the three-dimensional space to the Minkowski space-time.
Since $\mathcal{M} = \rmd\mathcal{P}$ is a 2-form, we may consider a \emph{current}
(a differential form with hyper-function coefficients) supported on a
two-dimensional surface.  
A pair of two-dimensional surfaces can link in the four-dimensional space-time.

The reason why the link of surfaces yields a topological constraint
on loops (vortex-filaments) is because such surfaces can be chosen as \emph{orbits} of loops.
The conventional vortex-filaments are, then, the temporal cross-sections of the surfaces.
When we consider a dynamical process in space-time, represented by
a diffeomorphism $\mathcal{T}_U(s)$, 
geometrical manipulations caused by $\mathcal{T}_U(s)$ is constrained by the
causality (on any reference frame, $t$ cannot go to negative).
For instance, consider a one-dimensional space $\mathbb{R}$.  
Two points on the two-dimensional space-time ($x$-$t$ plane) do not link, 
and the spacial projections of them, $x$ and $x'$,
may exchange their positions on the space-axis without intersecting their orbits in the space-time.
However, such exchange is only possible if one particle goes to the direction of negative $t$.
Otherwise, $x$ and $x'$ cannot change their order without causing an intersection of
their orbits in the $x$-$t$ plane.

Similarly, a topological constraint on the link of surfaces in the four-dimensional space-time
causes a constraint on the link of their temporal cross-sections, loops,
if the surfaces are the orbits of the loops.

\subsection{Pure-state vorticity and filaments}
\label{subsec:pure-state}

To explore the basic topological constraint due to the relativistic helicity conservation, we consider
the link of vortex filaments or surfaces that are formally the $\delta$-measures,
supported on loops or surfaces, with a \emph{unit magnitude}. 
To formulate them on a rigorous mathematical footing, 
we invoke the the notion of \emph{pure states}.

Let us start with a well-known example of commutative $C^*$ ring.
A ``point'' on a compact space $X$ is equivalent to a pure state of the Banach algebra $C^0(X,\mathbb{C})$,
which is a linear form $\eta$ such that $\eta(ff^*)\geq 0$ and $\eta(1)=1$
(it is ``pure'' if $\eta$ is not a mixture $\eta=(\eta_1+\eta_2)/2$ of two distinct states).
By Gelfand's theorem, every point $\bm{\xi}\in X$ can be
represented by a pure state $\eta_{\bm{\xi}}$ such that $\eta_{\bm{\xi}}(f)=f(\bm{\xi})$
for every $f\in C^0(X,\mathbb{C})$.
Algebraically, $\eta_{\bm{\xi}}$ is the quotient of $C^0(X,\mathbb{C})$ by 
the maximum ideal $\mathcal{I}_{\bm{\xi}}$ generated by $\|\bm{x}-\bm{\xi}\|$.
We may represent a pure state by a $\delta$-function, i.e.,
$\eta_{\bm{\xi}}(f)=\int_X f(\bm{x}) \delta(\bm{x}-\bm{\xi})\rmd^n x$.
Here, we generalize ``points'' to $p$-chains.

\begin{definition}[pure sate]
\label{definition:pure_state}
Let $M$ be a smooth manifold of dimension $n$ (in the present application, $M$ is the four-dimensional Minkowski space-time),
and $\Omega\subset M$ be a $p$-dimensional connected null-boundary submanifold of class $C^1$.
Each $\Omega$ can be regarded as an equivalent of a 
\emph{pure-sate functional} $\eta_\Omega$ on the space $\wedge^p T^*M$ of continuous $p$-forms:
\[
\eta_\Omega:\, \omega 
\mapsto \int_\Omega \omega,
\]
which can be represented as
\[
\eta_\Omega(\omega) 
= \int_M \mathfrak{J}(\Omega)\wedge\omega
= \int_\Omega \omega
\]
with an $(n-p)$-dimensional $\delta$-measure 
$\mathfrak{J}(\Omega)=\wedge^{n-p}\delta(x^\mu-\xi^\mu)\rmd x^\mu $,
where $x^\mu$ are local coordinates, and 
\[
\mathrm{supp}\,\mathfrak{J}(\Omega) = \Omega = \{ \bm{x}\in\mathbb{R}^n;\, x^\mu=\xi^\mu \,(\mu=1,\cdots,n-p) \} .
\]
We call $\mathfrak{J}(\Omega)$ a \emph{pure state} $(n-p)$-form,
which is a member of the Hodge-dual space of $\wedge^p T^*M$.
\end{definition}

On compact manifolds, the pure state of a submanifold $\Omega$ 
is obtained as a limit of the Thom form of $\Omega$,
and the cohomology class of a pure state of $\Omega$ 
corresponds to the Poincar\'e dual of $\Omega$\,\cite{Poincare-dual}.
We also note that the duality between submanifolds and pure state parallels the duality 
between infinite chains in homologies with closed support (or Borel-Moore homology)
and cochains in cohomologies\,\cite{cochain}.

\begin{remark}[pure states in quantum theory]
\label{remark:pure-state}
The notion of pure state is commonly used in quantum theory, where the 
observables are self-adjoint operators,
and their duals are wave functions, members of a Hilbert space $V$.
The pure states, which are the ``points'' on the unit sphere of $V$, 
constitute the spectral resolution of the \emph{identity} on $V$,
The aforementioned $C^*$ ring of scalar functions (classical mechanical observables) 
is \emph{quantized} by replacing the ring by a commutative ring of unitary operators;
then, the classical state, a point on coordinates, is replaced by a wave function, a point in the function space $V$.
In Definition\,\ref{definition:pure_state}, observables are differential forms, 
and their duals are the Hodge-dual $(n-p)$-forms.
The notion of a point is, then, generalized to a $p$-chain.
\end{remark}

\begin{figure}[bt]
  \centering
  \includegraphics[width=0.5\textwidth]{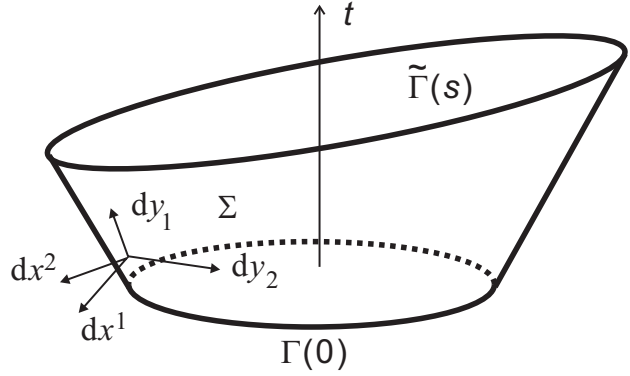}
  \caption{\label{fig:pure_state}
Pure-state on 2-forms in four-dimensional space-time, which is represented by a
2-chain $\Sigma$.  $\rmd y_1\wedge\rmd y_2$ is the tangential surface element on $\Sigma$
(i.e. $y_1$ and $y_2$ are the local coordinates on $\Sigma$),
and $\rmd x^1\wedge\rmd x^2$ is the normal 2-form 
(i.e. $*\rmd y_1\wedge\rmd y_2=\rmd x^1\wedge\rmd x^2$).
In these local coordinates, $\mathfrak{J}(\Sigma)=\delta_\Sigma \rmd x^1\wedge\rmd x^2$.
The initial-time cross-section of $\Sigma$ is denoted by $\Gamma(0)$.
An $s$-cross-sections of $\Sigma$ is denoted by $\tilde{\Gamma}(s)$.
}
\end{figure}

Let $\Sigma$ be a 2-dimensional connected null-boundary submanifold of class $C^1$ embedded in the four-dimensional space-time $M$ (see Fig.\,\ref{fig:pure_state}).
A pure-state functional $\eta_\Sigma(\omega) = \int_\Sigma \omega$ is 
represented by a pure-state 2-form $\mathfrak{J}(\Sigma)$ such that
\begin{equation}
\mathfrak{J}(\Sigma)=\delta_\Sigma * \rmd y_1\wedge \rmd y_2
=\delta_\Sigma  \mathfrak{m},
\quad \mathfrak{m}= \frac{1}{2} \mathfrak{m}_{\mu\nu}\rmd x^\mu\wedge\rmd x^\nu,
\label{2-form_current}
\end{equation}
where $\delta_\Sigma$ is the two-dimensional $\delta$-function supported on $\Sigma$,
$\rmd y_1\wedge \rmd y_2$ is a local two-dimensional surface element on $\Sigma$
($\rmd y_1$ and $\rmd y_2$ are tangent to $\Sigma$),
and $\mathfrak{m}=*\rmd y_1\wedge \rmd y_2$ is decomposed in terms of $\rmd x^\mu\wedge\rmd x^\nu$;
$\mathfrak{m}_{\mu\nu}$ is a certain antisymmetric tensor characterizing the surface $\Sigma$,
which we denote
\begin{equation}
\mathfrak{m}_{\mu\nu} =
\left( \begin{array}{cccc}
 0              &  \mathfrak{e}_1 &  \mathfrak{e}_2 &  \mathfrak{e}_3 \\
-\mathfrak{e}_1 &  0              & -\mathfrak{b}_3 &  \mathfrak{b}_2 \\
-\mathfrak{e}_2 &  \mathfrak{b}_3 &  0              & -\mathfrak{b}_1 \\
-\mathfrak{e}_3 & -\mathfrak{b}_2 &  \mathfrak{b}_1 &   0             \\
\end{array} \right) ,
\label{m-components_down}
\end{equation}
In Sec.\,\ref{subsec:orbit_of_loop}, we will identify $\Sigma$ as the
orbit of vortex filaments, and then, $\mathfrak{m}_{\mu\nu}$ of (\ref{m-components_down}) will be 
compared with the vorticity 2-form $\mathcal{M}_{\mu\nu}$ of (\ref{M-components_down}).

The $t$-cross-section of $\Sigma$ is
\begin{equation}
\Gamma(t) =  \Xi(t)\cap\Sigma ,
\label{t-filament}
\end{equation}
which is assumed to be a loop (1-cycle bounding a disk) in the space $\Xi(t)$.
Denoting by $\delta_{\Gamma(t)}$ the $\delta$-function supported on $\Gamma(t)$
(i.e. $\delta_{\Gamma(t)}=\delta_{\Xi(t)}\delta_\Sigma=\delta_{\Xi(t)\cap\Sigma}$), 
we define a ``$\bm{B}$-filament'' on a $t$-plane by
\begin{equation}
\mathfrak{j}_b(\Gamma(t)) = -\delta_{\Gamma(t)} * \mathfrak{b}^\dagger , 
\label{B-filament}
\end{equation}
where $\mathfrak{b}^\dagger =  \mathfrak{b}^k \rmd x_k$ with
$\mathfrak{b}^k=(1/2)\epsilon^{ijk}\mathfrak{m}_{ij}=-\mathfrak{b}_k$,
i.e. $\mathfrak{b}^\dagger$ is the Minkowski-dual of the magnetic-field-like 3-vector 
$\mathfrak{b}= \mathfrak{b}_k \rmd x^k$; 
notice the flip of the sign by the representation the 1-form $\mathfrak{b}^\dagger$ on the basis $\rmd x_k$.  
The projector $\rho_b(t) :\,\mathfrak{J}(\Sigma)\mapsto \mathfrak{j}_b(\Gamma(t))$ is written as
\begin{equation}
\rho_b(t) \mathfrak{J}(\Sigma)  =  - \delta_{\Xi(t)} \rmd x^0\wedge \mathfrak{J}(\Sigma) .
\label{t-projector_to_B}
\end{equation}
Notice that $\mathfrak{J}(\Sigma)$ has other three components including $\rmd x^0$ on a $t$-plane,
which we call an ``$\bm{E}$-filament'':
\begin{equation}
\mathfrak{j}_e(\Gamma(t)) = {-}\delta_{\Gamma(t)} * \mathfrak{e},
\label{E-filament}
\end{equation}
where $\mathfrak{e}= \mathfrak{e}_k \rmd x^k$ with
$\mathfrak{e}_k=(1/2)\epsilon_{ijk} *\mathfrak{m}_{ij}$.
The projector $\rho_e (t):\,\mathfrak{J}(\Sigma)\mapsto \mathfrak{j}_e(\Gamma(t))$ is
\begin{equation}
\rho_e(t) \mathfrak{J}(\Sigma)  =  \delta_{\Xi(t)} \rmd x_0\wedge * \mathfrak{J}(\Sigma).
\label{t-projector_to_E}
\end{equation}

In the relativistic formulation, the proper-time $s$-cross-section of $\Sigma$ is
more important: 
\begin{equation}
\tilde{\Gamma}(s) = \tilde{\Xi}(s)\cap\Sigma.
\label{s-filament}
\end{equation}
Generalizing (\ref{t-projector_to_B}), 
the ``relativistic $\bm{B}$-filament'' is a singular 3-form such that
\begin{equation}
{\mathfrak{j}}_b(\tilde{\Gamma}(s)) = \tilde{\rho}_b(s) \mathfrak{J}(\Sigma) 
= - \delta_{\tilde{\Xi}(s)} \mathcal{U}\wedge \mathfrak{J}(\Sigma) ,
\label{s-projector_to_B}
\end{equation}
where $\mathcal{U} = U_\mu \rmd x^\mu$.
Similarly, the relativistic $\bm{E}$-filament is defined as
(denoting $\mathcal{U}^\dagger = U^\mu \rmd x_\mu$)
\begin{equation}
{\mathfrak{j}}_e (\tilde{\Gamma}(s)) = \tilde{\rho}_e(s) \mathfrak{J}(\Sigma) 
= \delta_{\tilde{\Xi}(s)}  \mathcal{U}^\dagger\wedge * \mathfrak{J}(\Sigma).
\label{s-projector_to_E}
\end{equation}
Or, we may write 
\begin{equation}
\tilde{\rho}_e(s) \mathfrak{J}(\Sigma) 
=  {-}\delta_{\tilde{\Xi}(s)} * i_U \mathfrak{J}(\Sigma).
\label{s-projector_to_E'}
\end{equation}

Explicitly, we may write 
\begin{eqnarray*}
{\mathfrak{j}}_b (\tilde{\Gamma}(s)) &=& \delta_{\tilde{\Gamma}(s)} \tilde{\mathfrak{b}},
\\
{\mathfrak{j}}_e(\tilde{\Gamma}(s)) &=& {-}\delta_{\tilde{\Gamma}(s)} * \tilde{\mathfrak{e}},
\end{eqnarray*}
with a 3-form $\tilde{\mathfrak{b}}$ and a 1-form $\tilde{\mathfrak{e}}$ given by
(denoting $\bm{\mathfrak{b}}=\mathfrak{b}_k\rmd x^k$
and $\bm{\mathfrak{e}}=\mathfrak{e}_k\rmd x^k$)
\begin{eqnarray}
\begin{array}{lllll}
\tilde{\mathfrak{b}} &=& \tilde{\mathfrak{b}}_\mu \rmd x^{\nu_1}\wedge\rmd x^{\nu_2}\wedge\rmd x^{\nu_3}
&=&
\gamma(({-}\bm{v}/c)\cdot\bm{\mathfrak{b}},\, \bm{\mathfrak{b}} - (\bm{v}/c)\times\bm{\mathfrak{e}} ),
\label{Lorentz-transformation-B}
\\
\tilde{\mathfrak{e}} &=& \tilde{\mathfrak{e}}_\mu \rmd x^{\mu}
&=&
\gamma(-(\bm{v}/c)\cdot\bm{\mathfrak{e}},\, \bm{\mathfrak{e}} + (\bm{v}/c)\times\bm{\mathfrak{b}} ) .
\label{Lorentz-transformation-E}
\end{array}
\end{eqnarray}
We find that the 3-vector parts of $\tilde{\mathfrak{b}}$ and $\tilde{\mathfrak{e}}$ are, respectively, the
Lorentz transformations of the EM-like fields $\bm{\mathfrak{b}}$ and $\bm{\mathfrak{e}}$;
compare (\ref{M-components_down}) and (\ref{m-components_down}).

In general, neither ${\mathfrak{j}}_b(\tilde{\Gamma}(s))$ nor 
$\mathfrak{j}_b({\Gamma}(t))$
is a pure-state 3-form
(we use the lower-case $\mathfrak{j}$ to denote a non-pure-state singular filament).
Iff ${\mathfrak{j}}_e(\tilde{\Gamma}(s))=0$, then ${\mathfrak{j}}_b(\tilde{\Gamma}(s))$ is a pure state,
and it is the case when $\mathfrak{J}(\Sigma)$ 
belongs to the orthogonal complement of $\textrm{Ker}(\tilde{\rho}_b(s))$.
We will denote by ${\mathfrak{J}}_b(\tilde{\Gamma}(s))$ the pure-state $\bm{B}$-filament. 
Here we prove the following lemma:

\begin{lemma}
\label{lemma:pure-state-1}
Let $\Sigma\subset M$ be a 2-chain of class $C^1$,
whose proper-time cross-section $\tilde{\Gamma}(s)=\tilde{\Xi}(s)\cap\Sigma$ $(\forall s\in\mathbb{R})$ is
a single loop (1-cycle bounding a disk).
Suppose that $\mathfrak{J}(\Sigma)=\delta_\Sigma\mathfrak{m}$ is a pure-state 2-form.
Iff $\tilde{\rho}_e(s) \mathfrak{J}(\Sigma)=0$,
then $\tilde{\rho}_b(s) \mathfrak{J}(\Sigma)$
is a pure-state 3-from.
\end{lemma}

\begin{proof}
If $\mathfrak{J}(\Sigma)$ is a pure-state 2-form, we may write,
for every continuous 2-form $\omega$,
\begin{equation}
\int_M \mathfrak{J}(\Sigma)\wedge \omega =\int_\Sigma \omega .
\label{pure-state-filament-proof-1}
\end{equation}
Let us put $\omega = f_m \mathcal{U}\wedge a$, where $f_m$ is a scalar and $a$ is a 1-form.
Consider a sequence $f_m \rightarrow \delta_{\tilde{\Xi}(s)}$ ($m\rightarrow\infty$).
If $\mathcal{U}$ parallels $\Sigma$, i.e.
$i_U \mathfrak{J}(\Sigma) =0$
(implying that $\tilde{\rho}_e(s) \mathfrak{J}(\Sigma)=0$),
the right-hand side of (\ref{pure-state-filament-proof-1}) yields
\[
\int_\Sigma  f_m \mathcal{U}\wedge a \rightarrow 
\int_{\tilde{\Xi}(s)\cap\Sigma} a = \int_{\tilde{\Gamma}(s)} a.
\]
On the other hand, the left-hand side of (\ref{pure-state-filament-proof-1}) reads
\[
\int_M f_m \mathfrak{J}(\Sigma)\wedge \mathcal{U} \wedge a   
\rightarrow 
- \int_M \delta_{\tilde{\Xi}(s)}  \mathcal{U} \wedge \mathfrak{J}(\Sigma) \wedge a
= \int_M  [\tilde{\rho}_b(s) \mathfrak{J}(\Sigma)] \wedge a.
\]
Denoting $\mathfrak{J}_b(\tilde{\Gamma}(s))=\tilde{\rho}_b(s) \mathfrak{J}(\Sigma)$,
we may write, for every continuous 1-form $a$,
\[
\int_M \mathfrak{J}_b(\tilde{\Gamma}(s))\wedge a = \int_{\tilde{\Gamma}(s)} a,
\]
proving that $\mathfrak{J}_b(\tilde{\Gamma}(s))$ is a pure-state 3-form.

Next, we examine the normalization of the pure-states.  
Let us invoke the representation
$\mathfrak{J}(\Sigma)=\delta_\Sigma\mathfrak{m}_{\mu\nu} \rmd x^\mu\wedge \rmd x^\nu/2$ in terms of
the tensor (\ref{m-components_down}).
Then, we find
\begin{equation}
\mathfrak{m}\wedge * \mathfrak{m} = \sum_{k=1}^3 \left( \mathfrak{e}_k^2 - \mathfrak{b}_k^2 \right)
=|\bm{\mathfrak{e}}|^2 - |\bm{\mathfrak{b}}|^2 ,
\label{normalization-1}
\end{equation}
which must be unity for $\mathfrak{J}(\Sigma)$ to be a pure-state.
By (\ref{Lorentz-transformation-B}) and (\ref{Lorentz-transformation-E}), we obtain
\begin{eqnarray}
|\tilde{\mathfrak{b}}|^2 &=& \tilde{\mathfrak{b}}_\mu \tilde{\mathfrak{b}}^\mu
\nonumber \\
&=& \gamma^2 \left[ c^{-2}(\bm{v}\cdot\bm{{\mathfrak{b}}})^2-|\bm{\mathfrak{b}}-(\bm{v}/c)\times\bm{\mathfrak{e}}|^2 \right]
\nonumber \\
&=& |\bm{\mathfrak{e}}|^2 - |\bm{\mathfrak{b}}|^2 
+ |\tilde{\mathfrak{e}}|^2,
\label{normalization-2}
\end{eqnarray} 
where $|\tilde{\mathfrak{e}}|^2= \tilde{\mathfrak{e}}_\mu \tilde{\mathfrak{e}}^\mu$
(by writing (\ref{normalization-2}) as
$|\tilde{\mathfrak{b}}|^2 - |\tilde{\mathfrak{e}}|^2 = |\bm{\mathfrak{e}}|^2 - |\bm{\mathfrak{b}}|^2$,
which implies the Lorentz invariance of the Lagrangian, 
we notice that the sign changes because of the Minkowski metric).
By (\ref{normalization-1}), we have shown $|\bm{\mathfrak{e}}|^2 - |\bm{\mathfrak{b}}|^2 =1$.
Therefore, only if
$|\tilde{\mathfrak{e}}|^2=0$ (i.e. $\tilde{\rho}_e(s) \mathfrak{J}(\Sigma)=0$),
$|\tilde{\mathfrak{b}}|^2=1$ so that
$\tilde{\rho}_b\mathfrak{J}(\Sigma)=\delta_{\tilde{\Gamma}(s)}\tilde{\mathfrak{b}}$ is a 
pure state.
\qed
\end{proof}

\subsection{Orbit of vortex filaments}
\label{subsec:orbit_of_loop}

In Lemma\,\ref{lemma:pure-state-1}, we found that the projection of a pure-state 2-form $\mathfrak{J}(\Sigma)$
onto a proper-time plane $\tilde{\Xi}(s)$ yields
a pure-state 3-form $\mathfrak{J}_b(\tilde{\Gamma}(s))$ ($\bm{B}$-filament)
iff $i_U \mathfrak{J}(\Sigma)=0$, implying that the surface $\Sigma$ must be the \emph{orbit} of the
loop $\tilde{\Gamma}(s)$.
Here we construct $\mathfrak{J}(\Sigma)$ from $\mathfrak{J}_b(\tilde{\Gamma}(s))$
along the orbit.
The components ($\mathfrak{m}_{\mu\nu}$) of the pure-state 2-form $\mathfrak{J}(\Sigma)$
(which is the singular counterpart of the vorticity 2-form $\mathcal{M}=\rmd\mathcal{P}$) 
must be consistent with the dynamics equation.

As remarked in Sec.\,\ref{subsec:diffeo}, however, we may not solve the full set of equations to
determine $\mathfrak{J}(\Sigma)$.  
What we are going to construct are the elements of $\mathfrak{J}(\Sigma)$ that are
consistent with given $U$ and $\theta$
(whereas $U$ is related to $\mathcal{P}$, and is an unknown variable in the fluid/plasma equations;
$\theta$ must be consistent with $U$ through a relation $i_U\rmd\theta=0$).
We also note that the equation of motion must be generalized when we consider pure-state vorticities 
that are not regular functions obeying differential equations in the classical sense;
accordingly, the helicity conservation law must be reformulated to be amenable to the singular vorticities 
---this will be the task of Sec.\,\ref{subsec:toplogy*}.

Before the construction of $\mathfrak{J}(\Sigma)$, we generalize 
the equation of motion (\ref{EQM-dif-form}) in order to incorporate singular vorticities of pure states into the solutions. 
Let us first rewrite it as
\begin{equation}
i_U [\mathcal{M} + \mathcal{Q}] = 0,
\label{EQM-dif-form-2}
\end{equation}
where 
\begin{equation}
\mathcal{Q}= c^{-1} \mathcal{U}\wedge \rmd \theta 
\label{EQM-dif-form-2'}
\end{equation}
with $ \mathcal{U} = U^\mu \rmd x_\mu$.
By $i_U \mathcal{U} = 1$ and $i_U\rmd\theta=0$, it is evident that (\ref{EQM-dif-form}) and (\ref{EQM-dif-form-2}) are equivalent.
The \emph{adjoint form} of (\ref{EQM-dif-form-2}) is
\begin{equation}
\int_{\mathbb{R}^4} a\wedge *i_U [\mathcal{M} + \mathcal{Q} ] = 0.
\label{EQM-adjoint}
\end{equation}
In the present context, it is natural to consider the left-hand side of
(\ref{EQM-adjoint}) to be a linear form on continuous 1-forms.
If $\mathcal{M}$ satisfies (\ref{EQM-adjoint}) for every continuous 1-form $a$,
we say that $\mathcal{M}$ is a \emph{generalized solution} of the equation of motion (\ref{EQM-dif-form}).

As we have shown in Lemma\,\ref{lemma:support}, the support of the vorticity $\mathcal{M}=\rmd\mathcal{P}$
of a regular solution is co-moving with the four-dimensional flow $\mathcal{T}_U(s)$ 
(i.e. frozen into the fluid).
Therefore, a \emph{generalized solution} of the equation of motion (\ref{EQM-Lie-form}),
which may be regarded as some limit of regular solutions,
must be co-moving.
Let $\tilde{\Gamma}_0$ be an initial ($s=0$) single vortex filament 
which is a 1-cycle of class $C^1$ bounding a disk in $\tilde{\Xi}(0)$.
The \emph{orbit} of the single vortex filament is, 
denoting $\tilde{\Gamma}(s)= \mathcal{T}_U(s) \tilde{\Gamma}_0 $,
\begin{equation}
\Sigma = \bigcup_{s\in\mathbb{R}} \tilde{\Gamma}(s).
\label{orbit}
\end{equation}
Or, the family $\{\mathcal{T}_U(s)\tilde{\Gamma}(0)\}_{s \in \mathbb{R}}$ 
is the motion picture of the loop along $U$, 
which is a representation of a surface link\,\,\cite{motion_picture}.
We construct (using given $U$ and $Q$) a pure-state vortex $\mathfrak{J}(\Sigma)$ from an initial pure-state $\bm{B}$-filament $\mathfrak{J}_b(\tilde{\Gamma}_0) = \delta_{\tilde{\Gamma}_0} \tilde{\mathfrak{b}}$.
In Sec.\,\ref{subsec:pure-state}, we derived $\mathfrak{J}_b(\tilde{\Gamma}(s))$ from $\mathfrak{J}(\Sigma)$ 
by operating the projector $\tilde{\rho}_b(s)$.  
Here, we construct $\mathfrak{J}(\Sigma)$ from $\mathfrak{J}_b(\tilde{\Gamma}_0)$ by an
inverse map of $\tilde{\rho}_b(s)$, which, however, is not injective;
we must choose an appropriate inverse that is consistent with the (adjoint) equation of motion.

\begin{lemma}
\label{lemma:orbit}
Let $\Sigma$ be an orbit of a single loop $\tilde{\Gamma}(s)=\mathcal{T}_U(s)\tilde{\Gamma}_0$.
On $\tilde{\Gamma}(s)$ $(\forall s\in\mathbb{R})$, 
a pure-state $\bm{B}$-filament 
$\mathfrak{J}_b(\tilde{\Gamma}(s))=\delta_{\tilde{\Gamma}(s)}\tilde{\mathfrak{b}}$
and an $\bm{E}$-filament $\mathfrak{j}_e(\tilde{\Gamma}(s)) 
=-\delta_{\tilde{\Gamma}(s)} * \tilde{\mathfrak{e}}$ are given.
We define a pair of 2-forms, in the vicinity of $\tilde{\Gamma}(s)$, by
\begin{eqnarray}
\mathfrak{m}_b &=& -i_U \tilde{\mathfrak{b}} ,
\label{Gamma-orbit-b}
\\
\mathfrak{m}_e &=& \mathcal{U}\wedge \tilde{\mathfrak{e}} ,
\label{Gamma-orbit-e}
\end{eqnarray}
where $\mathcal{U}$ is the 1-form such that $i_U\mathcal{U}=1$.
\begin{enumerate}
\item
The $s$-plane projections of
the singular 2-forms $\delta_\Sigma \mathfrak{m}_b$ and $\delta_\Sigma \mathfrak{m}_e$ yield
\begin{eqnarray}
& &
\tilde{\rho}_b(s) \delta_\Sigma \mathfrak{m}_b = \mathfrak{J}_b(\tilde{\Gamma}(s)),
\quad \tilde{\rho}_e(s) \delta_\Sigma \mathfrak{m}_b =0,
\label{Gamma-orbit-b-projection}
\\
& &
\tilde{\rho}_e(s) \delta_\Sigma \mathfrak{m}_e = \mathfrak{j}_e(\tilde{\Gamma}(s)),
\quad \tilde{\rho}_b(s) \delta_\Sigma \mathfrak{m}_e = 0
\label{Gamma-orbit-e-projection}
\end{eqnarray}

\item
Let us put
\begin{eqnarray}
\mathcal{M} &=& \delta_\Sigma (\mathfrak{m}_b + \mathfrak{m}_e),
\label{Gamma-orbit-M}
\\
\mathcal{Q} &=& -\delta_\Sigma \mathfrak{m}_e.
\label{Gamma-orbit-Q}
\end{eqnarray}
If $\tilde{\mathfrak{e}}=-c^{-1}\rmd\theta$, then 
$\mathcal{M}$ together with $\mathcal{Q}$
satisfy the adjoint equation of motion (\ref{EQM-adjoint}).
\item
$\mathfrak{J}(\Sigma)= \delta_{\Sigma} \mathfrak{m}_b$ is a pure-state 2-form.
\end{enumerate}
\end{lemma}

\begin{proof}
The relations (\ref{Gamma-orbit-b-projection}) and (\ref{Gamma-orbit-e-projection}) 
follows directly from the definitions (\ref{Gamma-orbit-b}) and (\ref{Gamma-orbit-e}).
Evidently, $\tilde{\mathfrak{e}}=-c^{-1}\rmd\theta$ matches the definition (\ref{EQM-dif-form-2'}) of $\mathcal{Q}$.
We observe
\[
i_U \mathcal{M} = \delta_\Sigma i_U \mathfrak{m}_e = -i_U\mathcal{Q}.
\]
Hence, $i_U(\mathcal{M}+\mathcal{Q}) =0$, satisfying (\ref{EQM-adjoint}) for every continuous 1-form $a$.
By Lemma\,\ref{lemma:pure-state-1}, $\mathfrak{J}(\Sigma)$ is a pure state,
iff the $\bm{E}$-filament $\tilde{\rho}_e\mathfrak{J}(\Sigma)=0$.
By (\ref{Gamma-orbit-b-projection}), $\mathfrak{J}(\Sigma)\in\textrm{Ker}(\tilde{\rho}_e)$,
thus, the third statement of this Lemma is proven.
\qed
\end{proof}

In the definition (\ref{Gamma-orbit-M}), the total vorticity $\mathcal{M}$ consists of two parts,
$\delta_\Sigma \mathfrak{m}_b \in\textrm{Ker}(\tilde{\rho}_e)$
and $\delta_\Sigma \mathfrak{m}_e \in\textrm{Ker}(\tilde{\rho}_b)$.
As we have shown, the first part (denoted by $\mathfrak{J}(\Sigma)$) is a pure state; 
hence, the total $\mathcal{M}$ is not a pure state, if the second part (denoted by $-\mathcal{Q}$) is non-zero,
i.e. $\tilde{\mathfrak{e}}=-c^{-1}\rmd\theta\neq 0$.
This is the reason why the conventional helicity conservation is broken in a relativistic fluid
(see Sec.\,\ref{subsec:semi-relativistic}).
However, the relativistic helicity does conserve;
we have yet to prove this fact for singular vorticities.

\subsection{Helicity conservation, circulation theorem and linking number}
\label{subsec:toplogy*}

By Lemma \ref{lemma:orbit}, we have constructed the vorticity
$\mathcal{M}= \delta_\Sigma (\mathfrak{m}_b + \mathfrak{m}_e)$
satisfying the adjoint equation of motion (\ref{EQM-adjoint}).
The support $\Sigma$ of the vorticity is the orbit of  
the vortex filament $\tilde{\Gamma}(s)$ that is transported by the diffeomorphism 
$\mathcal{T}_U(s)$.
Hence, different vortex filaments do not intersect in space-time;
given disjoint 1-cycles $\tilde{\Gamma}_{1}(s)$ and $\tilde{\Gamma}_{2}(s)$, 
the corresponding surfaces $\Sigma_{1}$ and $\Sigma_{2}$ do not intersect, thus their temporal ($s$ or even $t$) cross-sections conserve
the linking number.
To put this fact in the perspective of the helicity conservation law,
we need to re-formulate the helicity for the singular vorticity
(see Remark\,\ref{remark:generalize_Stokes_formula}).

Here we consider a pair of disjoint loops (1-cycles bounding disks) $\tilde{\Gamma}_{1}(s)$ and $\tilde{\Gamma}_{2}(s)$,
and their orbits
$\Sigma_{1}=\bigcup\tilde{\Gamma}_{1}(s)$
and $\Sigma_{2}=\bigcup\tilde{\Gamma}_{2}(s)$.
The total vorticity is the combination of twin vorticities: 
\begin{equation}
\mathcal{M}=\mathcal{M}_{1}+\mathcal{M}_{2}
= \delta_{\Sigma_{1}}(\mathfrak{m}_{b,1} + \mathfrak{m}_{e,1}) 
+ \delta_{\Sigma_{2}}(\mathfrak{m}_{b,2} + \mathfrak{m}_{e,2}),
\label{twin-vorticity}
\end{equation}
where $\mathfrak{m}_{b,\ell}$ and $\mathfrak{m}_{e,\ell}$, respectively,
stem from pure-state $\bm{B}$-filaments $\mathfrak{J}_b(\tilde{\Gamma}_{\ell}(s))$ and
$\bm{E}$-filaments $\mathfrak{j}_e(\tilde{\Gamma}_{\ell}(s))$ ($\ell=1,2$), as constructed in Lemma\,\ref{lemma:orbit}.

Let $V(s)$ be a co-moving temporally-thin volume,
i.e. $V(s)\subset \tilde{\Xi}(s)$ (see Sec.\,\ref{subsec:R-helicity}).
For an arbitrary continuous 1-form $a$, we obtain,
using (\ref{Gamma-orbit-b-projection}),
\begin{eqnarray}
\int_{V(s)} a\wedge \mathcal{M}
&=& \int_{V(s)\cap(\Sigma_{1}\cup\Sigma_{2})} a
\nonumber
\\
&=& \int_{\tilde{\Gamma}_{1}(s)} a + \int_{\tilde{\Gamma}_{2}(s)} a.
\label{circulation-0}
\end{eqnarray}
The final expression is nothing but the \emph{circulation} of the 1-form $a$
along the cycles $\tilde{\Gamma}_{1}(s)$ and $\tilde{\Gamma}_{2}(s)$.
To use the formula (\ref{circulation-0}) in order to evaluate the helicity
$\int_{V(s)}\mathcal{P}\wedge\mathcal{M}$,
we have to insert $\mathcal{P}$ into $a$,
and then, we have to relate $\mathcal{P}$ with $\mathcal{M}$ by
inverting the defining relation $\rmd\mathcal{P}=\mathcal{M}$;
let us formally write
\begin{equation}
\mathcal{P}=\mathcal{F}\mathcal{M}.
\label{LW-formal}
\end{equation}
The operator $\mathcal{F}$ will be
explicitly defined in Lemma\,\ref{lemma:LW-operator}.
Here, we remark that the vorticity $\mathcal{M}$ of (\ref{twin-vorticity})
consists of $\delta$-measures, thus $\mathcal{P}$ is not continuous at $\Sigma_\ell$.
This difficulty can be removed by decomposing $\mathcal{P}$ as
\begin{equation}
\mathcal{P}=\mathcal{P}_1 + \mathcal{P}_2 
= \mathcal{F}\mathcal{M}_1 + \mathcal{F}\mathcal{M}_2,
\label{LW-formal-decomposed}
\end{equation}
and putting the ``self-field helicity'' zero (see Remark\,\ref{remark:generalize_Stokes_formula}-2),
i.e.
\begin{equation}
\int_{V(s)} \mathcal{P}_\ell\wedge \mathcal{M}_\ell
= \int_{\tilde{\Gamma}_{\ell}(s)} \mathcal{P}_\ell =0
\quad (\ell= 1,2).
\label{self-field_helicity}
\end{equation}
Then, the relativistic helicity of the twin vorticity (\ref{twin-vorticity}) evaluates as
\begin{equation}
\mathfrak{C}(s) = \int_{V(s)} \mathcal{P}\wedge \mathcal{M}
= \int_{\tilde{\Gamma}_1(s)} \mathcal{P}_2
 + \int_{\tilde{\Gamma}_2(s)} \mathcal{P}_1.
\label{circulation-1}
\end{equation}

Let us make the operator $\mathcal{F}$ explicit.

\begin{lemma}
\label{lemma:LW-operator}
We denote $\delta = *\rmd*$, and $\square = \delta\rmd + \rmd\delta$ (d'Alembertian).
We invert $\square$ by the Li\'enard-Wiechert integral operator, which we denote by $\square^{-1}$.
In (\ref{LW-formal}), we can define
\begin{equation}
\mathcal{P}= \mathcal{F}\mathcal{M} = \square^{-1} \delta \mathcal{M}.
\label{LW-oprator}
\end{equation}
\end{lemma}

\begin{proof}
First we transform 
\begin{equation}
\mathcal{P} \mapsto \mathcal{P}'=\mathcal{P}-\rmd \varphi 
\label{gauge_transformation}
\end{equation}
with a scalar function $\varphi$ such that
$\delta\rmd\varphi= \square\varphi = \delta\mathcal{P}$.
Operating $\delta$ on the both sides of 
$\rmd\mathcal{P} = \rmd\mathcal{P}'=\mathcal{M}$, we obtain
$\delta\rmd\mathcal{P}' = \delta\mathcal{M}$.
The left-hand side reads $\square\mathcal{P}'$, since $\delta\mathcal{P}'=0$.  
We obtain
\begin{equation}
\mathcal{P}' = \mathcal{F}'\mathcal{M}= \square^{-1} \delta \mathcal{M}.
\label{LW-1}
\end{equation}
Transforming back to $\mathcal{P}$, we obtain
$\mathcal{P} = \square^{-1} \delta \mathcal{M} + \rmd \square^{-1}\delta\mathcal{P}$,
thus we may write
\begin{equation}
\mathcal{P} 
=\mathcal{F}\mathcal{M}
= (1-\rmd \square^{-1}\delta)^{-1} \square^{-1} \delta \mathcal{M}.
\label{LW-2}
\end{equation}
Since $\int_{V(s)} \rmd\varphi \wedge \mathcal{M}=0$
for every $\mathcal{M}=\rmd\mathcal{P}$ such that $\mathcal{M}=0$ at $\partial V(s)$,
we may replace $\mathcal{F}$ by $\mathcal{F}'=\square^{-1} \delta$ in (\ref{LW-formal}).
\qed
\end{proof}

Now the following conclusions are readily deducible:

\begin{theorem}[link in Minkowski space-time]
\label{theorem:singular_vortex}
Let $\mathcal{M}=\mathcal{M}_1+\mathcal{M}_2$ be a twin vortex generated by
a pair $(\ell=1,2)$ of pure-state $\bm{B}$-filaments $\mathfrak{J}_b(\tilde{\Gamma}_\ell(s))$,
and $\bm{E}$-filaments $\mathfrak{j}_e(\tilde{\Gamma}_\ell(s))= \delta_{\tilde{\Gamma}_\ell(s)} * c^{-1}\rmd \theta$ supported on co-moving loops $\tilde{\Gamma}_\ell(s)$.
\begin{enumerate}
\item
The relativistic helicity
\begin{equation}
\mathfrak{C}(s) = \int_{\tilde{\Gamma}_1(s)} \mathcal{F}\mathcal{M}_2
 + \int_{\tilde{\Gamma}_2(s)} \mathcal{F}\mathcal{M}_1.
\label{circulation-2}
\end{equation}
is a constant of motion.
\item
The constant $\mathfrak{C}(s)/2$ is the linking number ${L}(\tilde{\Gamma}_1(s),\tilde{\Gamma}_2(s))$,
which may be represented as (generalizing the Gauss integral)
\begin{equation}
{L}(\tilde{\Gamma}_1(s),\tilde{\Gamma}_2(s))
= \int\mathcal{F}\mathcal{M}_2\wedge \mathfrak{J}_b(\tilde{\Gamma}_1(s))
= \int\mathcal{F}\mathcal{M}_1\wedge \mathfrak{J}_b(\tilde{\Gamma}_2(s)).
\label{circulation-3}
\end{equation}
\end{enumerate}
\end{theorem}

\begin{proof}
From the foregoing derivation, it is clear that (\ref{circulation-2}) is an
appropriate expression of the relativistic helicity generalizing (\ref{4Dhelicity}).
Since, $\mathcal{P}_1 = \mathcal{F}\mathcal{M}_1$ and $\mathcal{P}_2=\mathcal{F}\mathcal{M}_2$
are smooth (holomorphic) functions in the vicinities of $\tilde{\Gamma}_2(s)$ and $\tilde{\Gamma}_1(s)$,
respectively, and satisfies the equation of motion (\ref{EQM-Lie-form}) (Lemma\,\ref{lemma:orbit}),
we can apply Lemma\,\ref{lemma:circulation_law} (circulation law) to prove the constancy of
$\mathfrak{C}(s)$.
Or, we can calculate directly as
\begin{eqnarray}
\frac{\rmd}{\rmd s} \int_{\tilde{\Gamma}_{\ell'}(s)}\mathcal{F}\mathcal{M}_\ell
&=& \int_{\tilde{\Gamma}_{\ell'}(s)} L_U (\mathcal{F}\mathcal{M}_\ell)
\nonumber
\\
&=& \int_{\tilde{\Gamma}_{\ell'}(s)} i_U \mathcal{M}_\ell
\nonumber
\\
&=& \int_{\tilde{\Gamma}_{\ell'}(s)} \tilde{\mathfrak{e}}_\ell 
= -c^{-1}\int_{\tilde{\Gamma}_{\ell'}(s)} \rmd\theta
= 0,
\label{circulation-4}
\end{eqnarray}
where $(\ell,\ell')=(1,2)$ or $(2,1)$.
For every loop $\tilde{\Gamma}$ bounding a disk $\tilde{\sigma}\subset\tilde{\Xi}(s)$,
we observe
\begin{eqnarray*}
\int_{\tilde{\Gamma}}\mathcal{F}\mathcal{M}_\ell
&=& \int_{\tilde{\sigma}}\mathcal{M}_\ell
\\
&=& \int_{\bigcup_s\tilde{\sigma}(s)}-\delta_{\tilde{\Xi}(s)}\mathcal{U}\wedge\mathcal{M}_\ell
= \int_{\bigcup_s\tilde{\sigma}(s)} \mathfrak{J}_b(\tilde{\Gamma}_\ell(s)),
\end{eqnarray*}
where $\bigcup_s\tilde{\sigma}(s)$ is the orbit of $\tilde{\sigma}(s)$.
Since each $\mathfrak{J}_b(\tilde{\Gamma}_\ell(s))$ is a pure state,
the right-hand side yields $\pm1$, 
iff the loops $\tilde{\Gamma}_\ell(s)$ and $\tilde{\Gamma}=\partial\tilde{\sigma}(s)$ link
in $\tilde{\Xi}(s)$
(the sign depends the orientations of the loops $\tilde{\Gamma}_\ell$ and $\tilde{\Gamma}$), i.e.
$\int_{\tilde{\Gamma}}\mathcal{F}\mathcal{M}_\ell=L(\tilde{\Gamma},\tilde{\Gamma}_\ell)$.
Hence, we conclude that $\mathfrak{C}(s)/2$ is the
linking number of loops contained in an $s$-plane
$\tilde{\Xi}(s)$.
\qed
\end{proof}

\begin{remark}[separation of pure-state vorticity]
\label{remark:Subtraction_of_baroclinic_term}
From (\ref{circulation-4}), it is evident that a non-exact $\bm{E}$-filament
in a baroclinic fluid
(i.e. $\tilde{\mathfrak{e}}_\ell = -c^{-1} T\rmd S$ with a temperature $T$ and an
entropy $S$; see Sec.\,\ref{subsec:field-tensor})
brings about a change in the helicity (or the circulation).
It is also clear that the helicity (or the circulation)
evaluated only by the pure-state part $\delta_\Sigma\mathfrak{m}_{b,\ell}$ of the
vorticity $\mathcal{M}_\ell$ is conserved even in a baroclinic fluid.
Conservation of such a reduced helicity (or a reduced circulation) has been 
noticed in NR formulations of fluid mechanics; see 
\cite{Eckart,Mobbs,Fukumoto}
\end{remark}

\section{Discussion}
\label{sec:discussion}

This work was given its motivation by the finding of non-conservation of helicity
(or circulation) in a relativistic fluid\,\cite{Mahajan2003,MahajanYoshida2010}.
The relativistic distortion of space-time, measured by $\nabla \gamma\neq 0$,
yields relativistic baroclinic effect on
a thermodynamically barotropic fluid, and violates the conservation of helicity.
Since the vorticity is dressed by a magnetic field in a high-energy charged fluid,
the breaking of the helicity constraint gives rise to a seed (cosmological) magnetic field.
The aim of this work was set to unearth an alternative, generalized conservation law 
that dictates a deeper topological constraint beneath the superficial (or reference-frame dependent) non-conservation.

We have introduced the relativistic (Lorentz invariant) helicity (\ref{4Dhelicity}), 
which is conserved (with respect to the proper time) in a thermodynamically barotropic fluid
(Theorem\,\ref{theorem:helicity_conservation}).
Considering a pair of pure-state vorticity filaments (relativistic $\bm{B}$-filaments),
we have shown that the helicity-conservation law means the constancy of the linking number
in the proper-time cross-section of space-time
(Theorem\,\ref{theorem:singular_vortex}).

As shown in Sec\,\ref{sec:R-helicity},
the semi-relativistic helicity $C(t)$ ceases to be
a constant of motion in a relativistic fluid with $\rmd\theta\neq0$.  
The reason why it can change is NOT because vortex loops change their link
(the linking number of the $t$-cross-sections of $\Sigma_1$ and $\Sigma_2$ does not change),
but is because the circulation changes on the loops.
In another word, a $\bm{B}$-filament on a $t$-plane ($\mathfrak{j}_b(\Gamma(t))$) 
is not a pure state; 
instead, the pure state is the relativistic $\bm{B}$-filament on an $s$-plane (Lemma\,\ref{lemma:orbit}).
Interestingly, however, $C(t)$ does conserve if $\rmd\theta=0$ (i.e. in a homentropic fluid);
$C(t)$ must be, then, a linking number on a $t$-plane.
To see how $C(t)$ conserves,
we may replace the co-moving volume $V(s)$ in the definition of $\mathfrak{C}(s)$ by
$V(t)=\mathcal{T}_u(t)V_0$, where $\mathcal{T}_u(t)$ is the diffeomorphism generated by
the reference-frame 4-vector $u = c^{-1}\rmd x/\rmd t$; see (\ref{NR_4-velocity}).
Then, $\mathfrak{C}(s)$ converts to the semi-relativistic $C(t)$,
and Theorem\,\ref{theorem:helicity_conservation} modifies to conclude $\rmd C(t)/\rmd t=0$;
to prove this,
we just replace $s$ by $t$ and $U$ by $u$ in the proof of Theorem\,\ref{theorem:helicity_conservation}
as well as in Lemma\,\ref{lemma:circulation_law} and Lemma\,\ref{lemma:support}.
It is evident that these replacements applies as far as $\rmd\theta=0$.
Since the modified Lemma\,\ref{lemma:circulation_law} shows the conservation of the
circulation on every loop $\Gamma(t)$ on the $t$-plane,
the constant $C(t)$ can be made to measure the linking number of twin vortex filaments.
However, we have to apply a different normalization of the $\bm{B}$-filaments on the $t$-plane;
we set $|\bm{\mathfrak{b}}_\ell|^2=1$, instead of $|\tilde{\mathfrak{b}}_\ell|^2=1$,
to let $\mathfrak{j}_b(\Gamma_\ell(t))$ be a pure state.
By (\ref{normalization-2}), these two different normalizations conflict with each other,
because $|\bm{\mathfrak{e}}_\ell|^2\neq0$ whenever $\bm{v}\neq0$;
hence two constants $C(t)$ and $\mathfrak{C}(s)$ have different values.
It is needless to say that $\mathfrak{j}_b(\Gamma_\ell(t))$ can remain as a pure state
only if $\rmd\theta=0$.

We end this paper with a short summary of helicities in different systems and their comparisons.
For the EM potential $A=A_\mu\rmd x^\mu$, 
the helicity density is the 3-form $\mathcal{K}= A\wedge(\rmd A)= -A_\nu F^{*\mu\nu}*\rmd x_\mu$
($F^{*\mu\nu}$ the dual of the Faraday tensor),
which may be viewed as a 4-current in the Minkowski space-time.
The 0-component $\mathcal{K}^0$ is the familiar magnetic helicity density $\bm{A}\cdot(\nabla\times\bm{A})$.
The divergence of the current $\rmd\mathcal{K}$ reads $(1/2)F_{\mu\nu}F^{*\mu\nu}=-2\bm{E}\cdot\bm{B}$ 
with the standard EM fields $\bm{E}$ and $\bm{B}$.
The total ``charge'' $\int_X \mathcal{K}^0\rmd^3x$, the magnetic helicity, is invariant 
if $\rmd\mathcal{K}=0$
(remember the discussion in Sec.\,\ref{subsec:semi-relativistic}).
For example, a \emph{null EM field}\,\cite{Bateman}, such that
$F_{\mu\nu}F^{\mu\nu} = F_{\mu\nu}F^{*\mu\nu} = 0$, propagates in the vacuum with conserving the helicity (see \cite{Kedia} for the examples of knotted EM fields in light).
Also in an ideal MHD system, $\bm{E}\cdot\bm{B}=0$, thus the 
conventional magnetic helicity conserves 
(hence, the link of magnetic field lines is invariant; cf.\,\cite{Pegoraro}).
When dressed by the fluid-mechanical momentum, however, the helicity density $\mathcal{K}=\mathcal{P}\wedge\rmd\mathcal{P}$
is no longer divergence-free (excepting the simplest homentropic fluid),
and thus, the total charge is not conserved in the relativistic regime
(Sec.\,\ref{subsec:semi-relativistic}).
In the non-relativistic limit ($\gamma=1$), however,
the heat term $T\rmd S=\rmd\theta$ of a barotropic 
fluid may be absorbed by the enthalpy term to modify the helicity density to be divergence-free 
(Remark\,\ref{remark:NR-barotropic});
by subtracting $2c^{-1}\theta\mathcal{B}^j$ from the spacial part $\mathcal{K}^j$, we define
a modified helicity density
\begin{eqnarray*}
\mathcal{K}'^{\mu} = \left(\bm{\mathcal{A}}\cdot\bm{\mathcal{B}},\, \mathcal{A}_0\bm{\mathcal{B}} - \bm{\mathcal{A}}\times\bm{\mathcal{B}} - 2c^{-1}\theta\bm{B}\right).
\end{eqnarray*}
Because the right-hand side of (\ref{e:div K}) may be written as 
$\nabla\cdot(2c^{-1}\theta\bm{\mathcal{B}})$, we obtain $\partial_\mu \mathcal{K}'^\mu=0$.
The closed 3-form $\mathcal{K}'$ is a \emph{Noether current}
pertinent to the \emph{relabeling symmetry} of the action~\cite{Noether1,Noether2,Noether3,Noether4}.
The 0-component $\mathcal{K}'^0$ is the Noether charge, and its
spatial ($t$-plane) integral is the conventional (non-relativistic) helicity $C$.
The relativistic effect ($\nabla\gamma\neq0$), however,
makes the right-hand side of (\ref{e:div K}) non-exact, thus we cannot introduce such a modified
divergence-free helicity density.
Yet, the relativistic helicity is
made invariant by integrating $\mathcal{K}$ on a co-moving domain,
because $L_U\mathcal{K}$ is an exact 3-form.
Whereas $\mathcal{K}$ is not divergence-free (because of the relativistic length contraction),
we find that the Lagrangian representation of the helicity density is divergence-free.
The Lagrangian coordinates labeling the position of each fluid element
have the redundancy (symmetry) of relabeling, and the consequent Noether current turns out to be the present relativistic helicity density after transforming into the Eulerian coordinates;
detailed analysis of the symmetry of the relativistic Lagrangian will be published elsewhere.

\acknowledgments
We acknowledge a debt to the Isaac Newton Institute for Mathematical Sciences, University of Cambridge;
this work was given a chance to start during the workshop \emph{Topological Fluid Dynamics}.
We are grateful to Professor K. Moffatt and Professor Y. Mitsumatus for their suggestions and discussions.
The work of ZY was partially supported by Grant-in-Aid for Scientific Research from
the Japanese Ministry of Education, Science and Culture No. 23224014.
The work of YK was supported by Grant-in-Aid for JSPS Fellows 241010.
The work of TY was partially supported by the JST CREST Program at Department of Mathematics,
Hokkaido University.



\begin{thebibliography}{99}

\bibitem{link}
Loops do not link in four-dimensional space.  
In general, we have the following theorem
which is readily deducible by the Thom transversality:
Let $A$ and $B$ be, respectively, $\mu$ and $\nu$-dimensional submanifolds embedded in 
$\mathbb{R}^n$.
If $n-(\mu+\nu)\geq 2$, then 
$B$ is contractible in $\mathbb{R}^n - A$ and 
$A$ is contractible in $\mathbb{R}^n - B$.
See M. W. Hirsch, \textit{Differential topology}, 
Graduate Texts in Mathematics No. 33 (Springer-Verlag, New York-Heidelberg, 1976). 


\bibitem{Woltjer}
L. Woltjer,
Proc. Natl. Acad. Sci. U.S.A. \textbf{44}, 489-491 (1958).


\bibitem{Moffatt}
H. K. Moffatt, \textit{Magnetic field generation in electrically conducting fluids},
(Cambridge University Press, Cambridge, 1978).

\bibitem{Moffatt-Ricca}
H. K. Moffatt and R. Ricca,
Proc. Roy. Soc. London A \textbf{439}, 411-429 (1992). 

\bibitem{Cantarella-Parsley}
J. C. Cantarella and J. Parsley,
J. Geometry and Phys. \textbf{60} (2010) 1127--1155.

\bibitem{Borromean}
D. DeTurck, H. Gluck, R. Komendarczyk, P. Melvin, C. Shonkwiler, and D. S. Vela-Vick,
J. Math. Phys. \textbf{54}, 013515 (2003).

\bibitem{LandauLifshitz}
L. D. Landau and E. M. Lifshitz, {\it Fluid Mechanics} (2nd Ed.): Vol. 6 of Course of Theoretical Physics (Butterworth-Heinemann, 1987).


\bibitem{MahajanYoshida2010}
S. M. Mahajan and Z. Yoshida,
Phys. Rev. Lett. \textbf{105}, 095005 (2010).

\bibitem{Mahajan2003}
S. M. Mahajan, Phys. Rev. Lett. {\bf 90}, 035001 (2003).



\bibitem{Synge}
J. L. Synge, 
Proc. London Math. Soc. sec 2 \textbf{43}, 376 (1937).

\bibitem{Newcomb}
W. A. Newcomb,
Ann. Phys. (N.Y.) \textbf{3}, 347 (1958).

\bibitem{Pegoraro}
F. Pegoraro,
Eur. Phys. Lett. \textbf{99}, 35001 (2012).


\bibitem{Poincare-dual}
R. Bott and L. W. Tu,
\textit{Differential forms in algebraic topology},
Graduate Texts in Mathematics No. 82 (Springer-Verlag, New York-Berlin, 1982).

\bibitem{cochain}
B. Iversen, 
\textit{Cohomology of sheaves}, Universitext (Springer-Verlag, Berlin, 1986). 

\bibitem{motion_picture}
S. Kamada,
\textit{Braid and knot theory in dimension four}, 
Mathematical Surveys and Monographs 95 (American Mathematical Society, Providence, 2002). 

\bibitem{Eckart}
C. Eckart, \textit{Hydrodynamics of oceans and atmospheres},
(Pergamon Press, London, 1960).

\bibitem{Mobbs}
S. D. Mobbs, J. Fluid Mech. {\bf 108}, 475 (1981).

\bibitem{Fukumoto}
Y. Fukumoto and H. Sakuma,
Procedia IUTAM \textbf{7}, 213 (2013).


\bibitem{Bateman}
H. Bateman, \textit{The mathematical analysis of electrical and optical wave-motion}
(Dover, New York, 1915).

\bibitem{Kedia}
H. Kedia, I. Bialynicki-Birula, D. Peralta-Salas, and W. T. M. Irvine,
Phys. Rev. Lett. \textbf{111}, 150404 (2013).

\bibitem{Noether1}
R. Salmon, Ann. Rev. Fluid Mech. {\bf 20}, 225 (1988).

\bibitem{Noether2}
A. Yahalom, J. Math. Phys. {\bf 36}, 1324 (1995).

\bibitem{Noether3}
N. Padhye and P. J. Morrison, Phys. Lett. A {\bf 219}, 287 (1996).

\bibitem{Noether4}
N. Padhye and P. J. Morrison, Plasma. Phys. Rep {\bf 22}, 960 (1996).



\end{thebibliography}
\end{document}